\documentclass[a4paper,12pt]{article}
\usepackage[utf8]{inputenc}
\usepackage{fullpage}
\usepackage{amsmath, amssymb, amstext, amsfonts, amsthm, bbm, array, enumerate,scalerel}
\usepackage{mathrsfs}
\usepackage{nicefrac}
\usepackage{xcolor}
\usepackage[english]{babel}
\usepackage{tikz}
\usepackage{stmaryrd}
\usepackage{tikz-3dplot}
\usetikzlibrary{patterns}
\usetikzlibrary{plotmarks}
\usepackage[round]{natbib}
\usepackage{enumitem}

\newtheorem{theorem}{Theorem}[section]
\newtheorem{lemma}[theorem]{Lemma}
\newtheorem{corollary}[theorem]{Corollary}
\newtheorem{proposition}[theorem]{Proposition}
\newtheorem{assumption}[theorem]{Assumption}

\theoremstyle{definition}
\newtheorem{definition}[theorem]{Definition}
\newtheorem{example}[theorem]{Example}

\newtheorem{remark}[theorem]{Remark}

\numberwithin{equation}{section}
\numberwithin{theorem}{section}

\newcommand{\E}{\mathbb{E}}

\newcommand{\Xcal}{{\mathcal X}}
\newcommand{\Acal}{{\mathcal A}}
\newcommand{\Bcal}{{\mathcal B}}

\newcommand{\s}{{\mathcal S}}
\newcommand{\Dcal}{{\mathcal D}}
\newcommand{\Gcal}{{\mathcal G}}
\newcommand{\Scal}{{\mathcal S}}

\DeclareMathOperator{\supp}{supp}
\newcommand{\Fcal}{{\mathcal F}}
\renewcommand{\P}{{\mathbb P}}
\newcommand{\Pol}{\textup{Pol}}
\newcommand{\fdot}{\,\cdot\,}

\newcommand{\R}{\mathbb{R}}
\newcommand{\Q}{\mathbb{Q}}
\newcommand{\N}{\mathbb{N}}

\newcommand{\diag}{\operatorname{diag}}

\RequirePackage[colorlinks,citecolor=blue,urlcolor=blue, linkcolor=blue]{hyperref}

\newcommand\mydots{\hbox to 1em{.\hss.\hss.}}

\begin{document}

\title{Measure-valued processes for energy markets}
\date{}

\author{Christa Cuchiero\thanks{Vienna University, Department of Statistics and Operations Research, Data Science @ Uni Vienna, Kolingasse 14-16, 1090 Wien, Austria, christa.cuchiero@univie.ac.at
}  \quad  Luca Di Persio \thanks{University of Verona, Department of Computer Science, Strada le Grazie 15, 37134 Verona, Italy, luca.dipersio@univr.it }
\quad Francesco Guida \thanks{University of Trento and University of Verona, Department of Mathematics, Via Sommarive 14, 38123 Povo, Italy, francesco.guida@unitn.it
}\\
  Sara Svaluto-Ferro \thanks{University of Verona, Department of Economics, Via Cantarane 24, 37129 
Verona, Italy, luca.dipersio@univr.it 
\newline
The first author gratefully acknowledges financial support 
through grant Y 1235 of the START-program.
}}
\date{}

\maketitle

\begin{abstract}
We introduce a framework that allows to employ (non-negative) measure-valued processes for energy market modeling, in particular for electricity and gas futures. Interpreting the process' spatial structure as time to maturity, we show how the Heath-Jarrow-Morton  approach (see~\cite{HJM:92}) can be translated to this framework, thus guaranteeing  arbitrage free modeling in  infinite dimensions. We derive an analog to the HJM-drift condition and  then treat in a Markovian setting existence of non-negative measure-valued diffusions that satisfy this condition. To analyze mathematically convenient classes we build on \cite{CGPS:21} and consider measure-valued polynomial and affine diffusions,
where we can precisely specify the diffusion part in terms of  continuous
functions satisfying certain admissibility conditions. For calibration purposes these functions
can then be parameterized by neural networks yielding measure-valued analogs of neural SPDEs.
By combining  Fourier
approaches or the moment formula with stochastic gradient descent methods, this then allows for tractable calibration procedures which we also test by way of example on market data.
We also sketch how measure-valued processes can be applied in the context of renewable energy production modeling.

\end{abstract}

\tableofcontents

\section{Introduction}

In this article we show how  to employ non-negative measure-valued processes for energy market modeling, in particular for electricity and gas futures. Before describing in detail our approach we start by discussing some important features of these markets, following \cite{BBK:08}. 

The most liquidly traded products in electricity, gas, and also temperature markets are future
contracts as well as options written on these. 
These future contracts
deliver the underlying commodity  over a certain period rather than at one fixed instance of time. By their nature,
these contracts are often also called swaps, since they represent an exchange between
fixed and variable commodity prices. We shall however employ the terms \emph{future} or \emph{future contract} throughout.

Our focus  lies on these future markets as they exhibit higher liquidity than the spot energy markets. Indeed, even though it is possible to trade power and gas commodities on the spot market, one usually faces  high  storage and transaction costs, which in turn explains the lack of liquidity on the spot market. As a consequence it is potentially more difficult to recognize and model features or characteristic patterns of the spot price behavior. Nevertheless stochastic  models for the spot prices are widely applied and also used to derive the dynamics of future prices based on no-arbitrage principles. 
Indeed, this is one of the two main approaches how to to set up 
tractable mathematical models for future contracts in power markets. The second approach consists in directly modeling  the complete forward curve by
applying a Heath-Jarrow-Morton (HJM) type approach (see \cite{HJM:92} and \cite{BK:14} in the context of energy markets). We shall adapt this second approach and model directly an analog of the forward curve, however with non-negative measure-valued processes rather than function-valued ones. In the sequel, we shall omit ``non-negative'' and only use ``measure'' for ``non-negative measure''.

There are many mathematical motivations to deal with this kind of processes. In general, measure-valued processes are often used for modeling dynamical systems for which the spatial structure plays a significant role. In the current setting,  time to maturity takes the role of the spatial structure. This is similar to common forward curve modeling via stochastic partial differential equation (SPDE) as for instance in \cite{BK:14} or \cite{BDL:21}. 
The  potential advantage of using measure-valued processes instead of function-valued
ones is that many spatial stochastic processes do not fall into the framework of SPDEs. In addition, it is often easier to establish existence of a measure-valued process rather than of an analogous SPDE, say in some Hilbert space, which would for instance correspond to its Lebesgue density, but which does not necessarily exist.

The decisive economic reason for using measure-valued processes in  electricity and gas modeling  however lies in the very nature of the future contracts, namely as products that deliver the underlying commodity always over a certain period instead of one fixed instance in time. 
To formulate this mathematically, denote the price of a future with delivery over the
 \emph{time interval} $(\tau_1, \tau_2]$ at time $0 \leq t \leq  \tau_1$  by $F (t, \tau_1 , \tau_2 )$. Then, following \cite{BBK:08}, $F (t, \tau_1 , \tau_2 )$ can be written as a weighted  integral of instantaneous forward prices $f(t,u)$ with delivery at \emph{one fixed  time}  $\tau_1 < u \leq \tau_2$, i.e.
\begin{align}\label{eq:future}
F (t, \tau_1 , \tau_2 )=\int_{\tau_1}^{\tau_2} w(u,\tau_1, \tau_2) f(t,u) du,
\end{align}
where $w(u,\tau_1, \tau_2) $ denotes some weight function.
The crucial motivation for measure-valued process now comes from the fact that
there is no trading with the instantaneous forwards $f(t,u)$ for obvious reasons.
Thus, rather than  using $f(t,u) du$ in the expression of the future prices, we can also use a measure. 

An additional motivation stems from the empirically well documented fact that energy spot prices can jump at predictable times, e.g.~due to maintenance works in the power grid or political decisions which currently influence these markets substantially. In the context of such stochastic discontinuities we refer to the pioneering work \cite{FGGS:20} in the setup of multiple yield curves. To illustrate why these predictable jumps lead to measure-valued forward processes, consider the simple example where the spot price $S$  exhibits a jump  at the deterministic time $t^*$. In this case we can write $S_u=\widetilde S_u+J\mathbbm{1}_{\{u\geq t^*\}}$ for some jump-size $J$ which is supposed to be $\mathcal{F}_{t^*}$-measurable\footnote{Here, $(\mathcal{F}_t)_{0\leq t \leq T}$ denotes some filtration supporting the spot price process and $T >0$ some finite time horizon.} and some process $\widetilde S$ that we suppose for simplicity to be continuous. As the instantaneous forward price is according to \cite[equation (1.6)]{BBK:08} given by
\[
f(t,u)=\mathbb{E}_{\Q }[S_{u}| \mathcal{F}_t]=
\mathbb{E}_{\Q }[\widetilde{S}_{u}+ J \mathbbm{1}_{\{u\geq t^*\}}| \mathcal{F}_t]
=\mathbb{E}_{\Q }[\widetilde{S}_{u}| \mathcal{F}_t]+ \E_\Q[J| \mathcal{F}_t] \mathbbm{1}_{\{u\geq t^*\}},
\] 
for each $t < t^* < T$ and $u \in [t, T]$,
this  implies that $u\mapsto f(t, u)$ is discontinuous at $u=t^*$. 
As forward curve modeling requires to take derivatives with respect to $u$, the corresponding SPDE contains a differential operator and we thus have to deal with distributional derivatives, leading naturally to measure-valued processes 
that can be treated with the analysis presented in this paper.
We also refer to \cite{AH:22} where stochastic discontinuities are also used as one motivation for term structure models based on \emph{cylindrical} measure-valued processes.

Due to these reasons one of the main goals of this article is to establish a 
\emph{Heath-Jarrow-Morton} (HJM) approach for measure-valued processes.
To this end we pass to the Musiela parametrization and thus consider time to maturity instead of time of maturity. 

More precisely, fix a finite time horizon $T$, consider  a  filtered probability space $(\Omega, \mathcal{F},$ $(\mathcal{F}_t)_{t \in [0,T]}, \mathbb{P})$, and a measure-valued process $(\mu_t)_{t \in [0,T]}$ 
supported on the compact interval $E:=[0,T]$. The measure
$ \mu_t$ will play the role of $f(t,t+x)dx$, meaning that if its Lebesgue density ${\mu_t(dx)}/{dx}$ existed, it would correspond to instantaneous forward prices at time $t$ with time to maturity $x$.
Then future prices at time $t$ given by \eqref{eq:future} become
\[
F (t, \tau_1 , \tau_2 )=\int_{(\tau_1-t, \tau_2-t]} w(t+x;\, \tau_1, \tau_2) \mu_t(dx), \quad t \in [0,\tau_1].
\]
Note that the integration domain does not include the left  boundary point $\tau_1-t$ as we suppose a delivery over the left-open interval $(\tau_1, \tau_2]$. Observe that the this choice permits to cumulatively add delivery periods, meaning that, 
$$F (t, \tau_1 , \tau_n )
=\sum_{k=1}^{n-1}\int_{(\tau_k-t,\tau_{k+1}-t]} w(t+x;\, \tau_1, \tau_n) \mu_t(dx),\qquad t\in[0,\tau_1]$$
for each $\tau_1<\cdots<\tau_n$.

The next step is now to specify  a suitable no-arbitrage condition. As long as the future contracts exist for potentially all delivery periods, it is natural to rely on no-arbitrage conditions for large financial markets allowing for a continuum of assets. In this respect the notion of \emph{no asymptotic free lunch with vanishing risk} (NAFLVR), as introduced in \cite{CKT:16}, qualifies as an appropriate and economically meaningful  condition.
Mathematically, 
in our setting, this is equivalent to the existence of an equivalent (local) martingale measure $\mathbb{Q} \sim \mathbb{P}$ under which the (discounted) price processes of these contracts are (local)  martingales. 
Modulo some technical conditions, this can then be translated to the following \emph{HJM-drift condition} on the measure-valued process $(\mu_t)_{t \in [0,T]}$ (see Theorem~\ref{th:main1}): if there exists an equivalent   measure $\mathbb{Q} \sim \mathbb{P}$ such that
\[
\langle \phi, \mu \rangle + \int_0^{\cdot} \langle \phi'_s ,\mu_s \rangle ds,
\]
is $\mathbb{Q}$-martingale for all appropriate  test functions $\phi$, then NAFLVR holds. Here, the brackets mean $\langle \phi, \mu \rangle= \int_E \phi(x) \mu(dx)$.
Note that this is a weak formulation to make sense out of 
$$
\text{``}d\mu_t(dx)=\frac{d}{dx} \mu_t(dx) dt+ dN_t(dx)\text{''},
$$
where $N$ denotes a measure-valued local martingale. Indeed, this would correspond to a (non-existing) strong formulation, well-known from function-valued forward curve modeling, see e.g.~\cite[Section 2]{BDL:21}.
Note that even though we here focus on energy markets, a similar framework can be used for term structure models in interest rate theory. The HJM-drift condition of course has to be adapted to guarantee that the bond prices are (local) martingales.

Clearly the HJM-drift condition restricts the choice of the measure-valued process since the drift part is completely determined, but we are free to specify the 
martingale part as long as we do not leave the state space of (non-negative) measures. To establish existence of measure-valued diffusion processes satisfying the drift condition, we rely on  the martingale problem formulation of \cite{EK:09} in locally compact and separable spaces which applies to our setting of measures on a compact set equipped with the weak-$*$-topology. Based on the positive maximum principle we then give sufficient conditions for the existence of such measure-valued diffusions (see Theorem~\ref{mainexistence}).
Additionally to these requirements we look for tractable specifications coming from the class of affine and polynomial processes as introduced in \cite{CGPS:21}, where we can precisely specify the diffusion part due to sufficient and (partly) necessary conditions on its form. This setup then allows for tractable pricing procedures via the moment formula well-known from polynomial processes (see e.g.~\cite{CKT:12}, \cite{FL:16}) and Fourier approaches.
This holds in particular true for the affine class as we obtain \emph{explicit solutions} for the Riccati PDEs.
These pricing methods can then of course be used for  calibration purposes. To this end we parameterize the function valued characteristics of polynomial models as neural networks to get \emph{neural measure-valued processes}, an analog to neural SDEs and SPDEs. In Section \ref{sec:cali} we exemplify such a calibration for the particular case of an affine measure-valued model which we calibrate to market call options  written on certain future contracts whose prices are obtained from the EEX German Power data 
(see \url{https://www.eex.com/en/market-data/power/options}).

Apart from modeling energy futures
via the described HJM approach, measure-valued processes also qualify for modeling quantities related to renewable energy production. In Section \ref{sec:wind} we exemplify this briefly by considering wind energy markets.

Finally, let us remark that the setting of cylindrical measures considered in \cite{AH:22} constitutes an interesting alternative  approach to term structure modeling. In contrast to our framework there cylindrical measure-valued processes arising from SPDEs  driven by some cylindrical Brownian motion are used to model the forwards. Note that being merely a cylindrical measure-valued process is a weaker requirement than being a measure-valued process. This weakening in turn allows to define SPDEs via cylindrical integrals, which would to be possible in the current measure-valued Markov process setup. 
One advantage of our approach is that we do not need Lipschitz conditions on the volatility function (imposed in  \cite{AH:22}), which is essential in order to accommodate affine and polynomial specifications. Another difference is that we consider concrete model specifications, in particular
conditions that guarantee to stay in the state-space of non-negative measures, while the focus in \cite{AH:22} lies rather on the abstract framework of cylindrical stochastic integration. Moreover, in the application to energy contracts, in \cite{AH:22}  the time \emph{of} maturity parametrization (in contrast to the Musiela parametrization) is considered, which does not require to specify a HJM drift condition. Note that we could do this as well by  specifying a measure-valued (local) martingale. The disadvantage of this approach is that -- in order to keep the interpretation as forward process -- the support of this measure-valued process should be $[t, T]$ and thus depends on the running time which is difficult to achieve. We therefore opted for the HJM approach with time \emph{to} maturity where the support of the measure-valued process is constant over time.

The remainder of the paper is organized as follows: in Section \ref{sec:notation} we introduce important notation used throughout the paper.
Section \ref{sec:market} is dedicated to define the large financial market setting and the no-arbitrage condition NAFLVR. In Section \ref{sec:HJM} we establish the HJM approach for measure-valued processes and derive the corresponding drift condition, which then implies NAFLVR. In Section \ref{sec:measurediff}
we tackle existence of measure-valued diffusions satisfying the HJM-drift condition, while Section \ref{sec:poly} is devoted to specify tractable examples of polynomial and affine type. In Section \ref{sec:options} we  show how this can be exploited for pricing of options written on future contracts, while in Section \ref{sec:cali} we present a calibration example based on a measure-valued neural affine process.
Finally Section \ref{sec:wind} sketches applications of measure-valued processes in wind energy markets.

\subsection{Notation and basic definitions}\label{sec:notation}
Throughout this article, we denote by $E$ the compact interval $[0,T]$ for some finite time horizon $T >0$ and  endow it with its Borel $\sigma$-algebra.  Moreover, we shall use the following notation.

\begin{itemize}
\item $M_+(E)$ denotes the set of finite non-negative measures on $E$, 
and $M(E)=M_+(E)-M_+(E)$ the space of signed measures of the form $\nu_+-\nu_-$ with $\nu_+,\nu_-\in M_+(E)$. We usually leave ``non-negative'' away and just say finite measures for  $M_+(E)$.
Both $M(E)$ and $M_+(E)$ are equipped with the topology of weak convergence, which turns $M_+(E)$ into a Polish space. 
Recall that $M_+(E)$ is locally compact since $E$ is compact (see, e.g.~\cite[Remark 1.2.3]{L:70}).  We denote its one-point compactification by $M^{\mathfrak{\Delta}}_+(E):= M_+(E)\cup\{\mathfrak{\Delta}\}$ and  identify  it with measures of infinite mass. 

\item $C(E)$ denotes the space of  continuous real functions on $E$ equipped with the topology of uniform convergence. We denote by $\|\fdot\|$ the supremum norm. Similary $C(M_+(E))$, $C_0(M_+(E))$, $C_c(M_+(E))$ denote, respectively, the spaces of  continuous, vanishing at infinity, and compactly supported real functions on $M_+(E)$.

\item For $k \in \mathbb{N} $, $C^k(E)$ denotes the restriction of $k$-times continuously differentiable functions on $\mathbb{R}$, denoted by $C^k(\mathbb{R})$, to $E$. Similary, $C^{\infty}(E)$ is the restriction of smooth functions on $\mathbb{R}$, denoted by $C^{\infty}(\mathbb{R})$, to $E$.

\item $\widehat C(E^k)$ is the closed subspace of $C(E^k)$ consisting of symmetric functions $f$, i.e., 
$f(x_1,\ldots,x_k)=f(x_{\sigma(1)},\ldots,x_{\sigma(k)})$ for all $\sigma\in \Sigma_k$, the permutation group of $k$ elements.
 For any $g\in\widehat C(E^k), h\in \widehat C (E^\ell)$ we denote by
$g\otimes h\in\widehat C(E^{k+\ell})$ the symmetric tensor product, given by

$$
\begin{aligned}
(g\otimes h) & (x_1,\ldots,x_{k+\ell}) \\
&= \frac{1}{(k+\ell)!} \sum_{\sigma\in\Sigma_{k+\ell}} g\big(x_{\sigma(1)},\ldots,x_{\sigma(k)}\big)h\big(x_{\sigma(k+1)},\ldots,x_{\sigma(k+\ell)}\big).
\end{aligned}
$$
 We emphasize that only symmetric tensor products are used in this paper.
\item Throughout the paper we let
$$
D \subseteq C (E)
$$
be a dense linear subspace containing the constant function 1. We set $D^{\otimes 2}=D\otimes D:=\text{span}\{g\otimes g\ :\ g\in D\}$. This generalizes  to $D^{\otimes k}$ for $k > 2$.

\end{itemize}

Two key notions that we shall often use are the \emph{positive maximum principle,} and the \emph{positive minimum principle}  for certain linear operators.
\begin{definition} 
Fix a Polish space $\Xcal$ and a subset $\s\subseteq \Xcal$. Moreover, let $D(\Acal)$ be a linear subspace of $C(\Xcal)$.
\begin{itemize}
\item 
An operator $\Acal$ 
with domain $D(\Acal)$  is said to satisfy the positive maximum principle on $\s$ if 
$$\text{$f\in D(\Acal)$, $x\in \s$, $\sup_\s f=f(x)\geq0\quad$ implies $\quad\Acal f(x)\leq0$.}$$

\item An operator $\Acal$ 
with domain $D(\Acal)$ is said to satisfy the positive minimum principle on $\s$ if 
\[
0 \leq g  \in D(\Acal),\, x \in \s,\, \inf_{\s}g=g(x)=0\quad \text{implies} \quad \mathcal{A}g(x)\geq 0.
\]
\end{itemize}
\end{definition}

\begin{remark}
Note that the positive maximum principle implies the positive minimum principle. Indeed, let $ g\geq 0$ with $\inf_\s g=g(x)=0$ and set $f=-g$ (which is possible since $D$ is a linear space). Then the positive maximum principle yields $\mathcal{A}f(x)=-\mathcal{A}g(x)\leq 0$.

Moreover, let $\Acal$ be an operator with domain $D(\Acal)$ satisfying the positive minimum principle on $\Scal$ and suppose that $1\in D(\Acal)$. 
Then $\Acal g =\Bcal g+mg$ for some map $m$ and some operator $\Bcal$ satisfying $\Bcal1=0$ and the positive maximum principle on $\Scal$. $\Bcal$ and $m$ can be explicitly constructed by setting $\Bcal g(x)=\Acal(g-g(x))(x)$ and $m(x):=\Acal 1(x)$.
\end{remark}

These notions shall play an important role to establish existence of measure-valued diffusions. Indeed, the positive maximum principle (combined with conservativity) is essentially equivalent to the existence of an $\Scal$-valued solution to the martingale problem for $\Acal$, see \cite[Theorem~4.5.4]{EK:09}. Here it is crucial that $\Scal$ is locally compact.  In our setting we shall apply this to $\Scal= M_+(E)$ which is locally compact.

The positive minimum principle will be used 
in the context of the HJM-drift condition. Note here that for generators of strongly continuous semigroups on $C(\mathcal{X})$ the  positive minimum principle is equivalent to generating a \emph{positive} semigroup, if $\mathcal{X}$ is compact, see \cite[Theorem~B.II.1.6]{arendt:86}.

\section{The large financial market and the no-arbitrage condition}\label{sec:market}

Fix a finite time horizon $T> 0$ and let $(\Omega,\mathcal{F},(\mathcal{F}_t)_{t\in [0,T]},\mathbb{P})$ be a filtered probability space, with an increasing and right-continuous filtration $(\mathcal{F}_t)_{t\in [0,T]}$. We suppose that the basic traded instruments in our market consist of 
\begin{itemize}
\item
a continuum of future contracts with delivery period over a time interval
$(\tau_1 , \tau_2 ]$ with $0\leq \tau_1 < \tau_2 \leq T$, whose prices at time $t \leq \tau_1$ are denoted by $F (t, \tau_1, \tau_2 )$, i.e.
$$
\lbrace F(t,\tau_1,\tau_2)_{t \in [0, \tau_1]} |\; 0\leq \tau_1 < \tau_2 \leq T\rbrace $$
are the traded energy products;

\item a bank account $(B_t)_{t \in [0,T]}\equiv 1$ corresponding to a constant risk-free rate  $r\equiv 0$ and chosen to be the num\'eraire.
\end{itemize}
Observe that due to the chosen risk-free rate the involved quantities are already discounted.
As a next step we now introduce conditions guaranteeing absence of arbitrage.
Due to the large financial market nature, we consider the no-arbitrage condition \emph{no asymptotic free lunch with vanishing risk} (NAFLVR) introduced in  \cite{CKT:16}.\footnote{Observe that the current setup can be embedded in Example~2.2 in \cite{CKT:16} by considering an index set $I =[0,T] \times [0,T]$ for the delivery periods from $\tau_1$ to $\tau_2$.}
In words, (NAFLVR) requires that there is no sequence of terminal payoffs of so-called admissible generalized strategies that
\begin{itemize}
    \item converges almost surely to a non-negative random variable which is strictly positive with positive probability,
    \item while their negative parts tend to 0 in $L^\infty$.
\end{itemize}
The \emph{fundamental theorem of asset pricing} in this context then asserts that NAFLVR is equivalent to the existence of a so-called equivalent separating measure $\mathbb{Q}$; see \cite[Theorem~3.2]{CKT:16}. 
In the case where all $F(\fdot, \tau_1, \tau_2)$ are locally bounded processes, equivalent separating measures correspond to equivalent local martingale measures. As we shall deal with continuous processes, we therefore have the following assumption.

\begin{assumption}[NAFLVR]\label{ass:NAFLVR}
There exists an equivalent measure $\mathbb{Q} \sim \mathbb{P}$ such that all elements in
 $
\lbrace F(t,\tau_1,\tau_2)_{t \in [0, \tau_1]} |\; 0\leq \tau_1 < \tau_2 \leq T\rbrace $
 are local martingales.
\end{assumption}

\section{A Heath-Jarrow-Morton approach for measure-va\-lued processes} \label{sec:HJM}

In this section we adapt the Heath-Jarrow-Morton (HJM) approach (see \cite{HJM:92}) to the current measure-valued setting which will in turn guarantee NAFLVR.

We start by briefly summarizing the classical HJM approach in a function-valued setup, usually in some Hilbert space, to model forward prices; see e.g.~\cite{BK:14}.
Recall that  $f(t,u)$ stands for the price at time $t$ of an instantaneous forward contract with delivery exactly at time $u$ with $0\leq t \leq u\leq T$.
Then, as outlined in the introduction, future prices $F(t, \tau_1, \tau_2)$ are given by \eqref{eq:future}. Hence, to guarantee Assumption~\ref{ass:NAFLVR} it suffices to assume
$$
(f(t,u))_{t\in[0,u]}=(\overline M_t(u))_{t\in[0,u]}, \quad \text{for all } 0 < u \leq T, $$
with $\overline M$ some (function-valued local) martingale under some equivalent measure $\mathbb{Q}$. Supposing sufficient regularity of $f$ and passing to the Musiela parametrization via 
$$\mathfrak{f}(t,x):=f(t,t+x),$$
the (local) martingale condition translates to the  SPDE
\begin{equation}\label{eq:energyspde}
\mathfrak{f}(t,x)-\int_0^t\frac d {dx}\mathfrak{f}(s,x) ds=M_t(x),
\end{equation}
where $dM_t=d\overline M_t(s+\fdot)|_{s=t}$ and hence $M$ is again a (function-valued local)-martingale.

Since the forwards cannot be traded and stochastic discontinuities in the spot prices give rise to measure-valued forward prices as explained in the introduction, the goal is now to replace the function-valued process $(\mathfrak{f}(t,\fdot))_{t \in [0,T]}$ by a measure-valued one $(\mu_t(\fdot))_{t \in [0,T]}$ supported on $E=[0,T]$. This means to model  future prices  as
\begin{equation}
\begin{split}\label{eq:future1}
F(t, \tau_1, \tau_2)&= \int_{(\tau_1-t, \tau_2-t]} w(t+x; \, \tau_1, \tau_2) \mu_t(dx)\\&=\int_E  w(t+x; \, \tau_1, \tau_2)\mathbbm{1}_{(\tau_1, \tau_2]}(t+x) \mu_t(dx),\quad t \in [0,\tau_1]
\end{split}
\end{equation}
instead of \eqref{eq:future}. Recall that by our delivery convention just the right boundary point is included in the integration domain. We shall also suppose that $ w(\fdot; \, \tau_1, \tau _2) \in C^\infty(\R)$.

As a first step let us introduce two functions sets. For functions $\Phi$ of two variables we denote by $t$ the first argument and by $x$ the second one. We thus use $\frac d {dt}\Phi$ and $\frac d {dx}\Phi$ accordingly. The two sets are
\begin{equation}\label{eq:D1}
    \begin{aligned}
D_1&=\{ \phi \in C^{\infty}(\mathbb{R}) \colon \phi'(0)=0 \}, \\
D_2(\tau_1)&=\{\Phi \in C^{\infty}(\mathbb{R}^2) \colon \frac d {dx} \Phi(t,0)=0 \text{ for all $t \in [0,\tau_1]$}\},
\end{aligned}
\end{equation}
which correspond to $C^{\infty}$-functions with vanishing $x$-derivative  at $x=0$.
  
We are now ready to provide sufficient conditions on $(\mu_t)_{t\in[0,T]}$ guaranteeing NAFLVR in sense of Assumption~\ref{ass:NAFLVR}. The result is based on the measure-valued formulation  of the HJM drift condition \eqref{eq:energyspde}, given in \eqref{eq:martmeas}. To see the analogy, observe that by integration by parts for each $\phi\in D_1$  the martingale $M$ introduced in \eqref{eq:energyspde} satisfies
\begin{align*}
    M_t
&=\int_E \phi(x)\mathfrak{f}(t,x)dx
-\int_0^t\int_E \phi(x)\frac d {dx}\mathfrak{f}(s,x)dxds\\
&=\int_E \phi(x)\mathfrak{f}(t,x)dx+\int_0^t\int_E \phi'(x)\mathfrak{f}(s,x)dxds.
\end{align*}

\begin{theorem}\label{th:main1}
Fix an equivalent measure $\mathbb{Q} \sim \mathbb{P}$ and
suppose that $(\mu_t)_{t \in[0,T]}$ is a $M_+(E)$-valued process such that 
\begin{equation}\label{eq:finitesup}
\mathbb{E}_{\mathbb{Q}}[\sup_{t \in [0,T]} \mu_t(E)] < \infty.
\end{equation}
Assume that for all $\phi \in D_1$
\begin{align}\label{eq:martmeas}
\langle \phi, \mu \rangle + \int_0^{\cdot} \langle \phi', \mu_s \rangle ds
\end{align}
is a  $\mathbb{Q}$-local martingale\footnote{Note that due to \eqref{eq:finitesup},  $\Q$-local martingality directly implies $\Q$-martingality.} on $[0,T]$.
Then the market satisfies NAFLVR in the sense of Assumption~\ref{ass:NAFLVR}. 
\end{theorem}

\begin{proof}
We prove that $(F(t, \tau_1, \tau_2))_{t \in [0,\tau_1]}$ are $\mathbb{Q}$-martingales (instead of only local martingales).
Note first that by the definition in \eqref{eq:future1}
\[
F(t, \tau_1, \tau_2)=
\langle \Phi(t,\fdot), \mu_t \rangle,\quad\text{for}\quad \Phi(t,x):=w(t + x; \, \tau_1, \tau_2)\mathbbm{1}_{(\tau_1, \tau_2]}(t+x),
\] 
for all $t \in [0,\tau_1]$.
We now aim to apply 
Lemma~\ref{lem:mart}.
Let $h_n\in C^{\infty}(\R)$ be a bounded sequence of  functions approximating $\mathbbm{1}_{(\tau_1, \tau_2]}$ pointwise such that $h_n(u)=0$ for each $u\in[0,\tau_1]$. 
Define
$$
F^n(t, \tau_1, \tau_2):=\langle \Phi_n(t,\fdot), \mu_t \rangle,\quad\text{for}\quad\Phi_n(t,x):=w(t + x; \, \tau_1, \tau_2)h_n(t+x).
$$
Note that $F^n(\fdot, \tau_1, \tau_2)\in C^{\infty} (\mathbb{R})$  and 
$$
\frac d {dx} \Phi_n(t,0)=\frac d {dt} w(t ; \, \tau_1, \tau_2)h_n(t)+w(t ; \, \tau_1, \tau_2)h'_n(t)= 0, \quad t \in [0, \tau_1],
$$
as $h_n(t)=0$ and $h_n'(t)=0$ for all $t \in [0, \tau_1]$. This shows that $\Phi_n\in D_2(\tau_1)$.
Moreover, $\frac d {dx} \Phi_n(t,x)=\frac d {dt} \Phi_n(t,x)$ for all $(t,x) \in [0, \tau_1] \times E$. Therefore we have for all $t \in [0,\tau_1]$
\begin{align*}
F^n(t, \tau_1, \tau_2)
=\langle \Phi_n(t,\fdot), \mu_t \rangle 
 =
\langle  \Phi_n(t,\fdot), \mu_t \rangle  
+ \int_0^t( \langle \frac d {dx}   \Phi_n(s,\fdot), \mu_s \rangle - \langle \frac d {dt}  \Phi_n(s,\fdot), \mu_s \rangle )ds,
\end{align*}
and Lemma~\ref{lem:mart} thus implies the martingale property of $(F^n(t, \tau_1, \tau_2))_{t \in [0,\tau_1]}$.
By dominated convergence $F^n(t, \tau_1, \tau_2)$ converges to $F(t, \tau_1, \tau_2)$ and
$$F^n(t, \tau_1, \tau_2) \leq \sup_{u\in[\tau_1,\tau_2]}|w(u;\tau_1,\tau_2)| \mu_t(E),$$
which is integrable with respect to $\mathbb{Q}$. Putting this together yields
\begin{align*}
\mathbb{E}_{\mathbb{Q}}[F(t, \tau_1, \tau_2) | \mathcal{F}_s]&=\mathbb{E}_{\mathbb{Q}}[\lim_{n \to \infty} F^n(t, \tau_1, \tau_2) | \mathcal{F}_s]=\lim_{n \to \infty}  \mathbb{E}_{\mathbb{Q}}[F^n(t, \tau_1, \tau_2) | \mathcal{F}_s]
= F(s, \tau_1, \tau_2),
\end{align*}
whence the martingale property of $(F(t, \tau_1, \tau_2))_{t \in [0, \tau_1]}$ under $\mathbb{Q}$.
\end{proof}

\begin{remark} 
Note that the setting of Theorem~\ref{th:main1} includes the most commonly  used weight function given by
\[
w(t+x; \, \tau_1, \tau_2)= \frac{1}{\tau_2- \tau_1},
\]
see \cite[Chapter 4]{BBK:08}. 
\end{remark}

\begin{remark}
In the setup considered so far the spot price did not play role. We here discuss how it can be deduced in the current framework. Fix some $t \in E$ and denote the spot price by $S_t$.
Note that in accordance with \cite[Chapter 4]{BBK:08}, $S_t$ would correspond to $$S_t=\lim_{\tau_1, \tau_2 \to t }F(t, \tau_1, \tau_2).$$ Using the above weight function this yields
\[
\lim_{\tau_1, \tau_2 \to t }F(t, \tau_1, \tau_2)= \lim_{\tau \to 0} \frac{1}{\tau}  \mu_t((0,\tau]).
\]
By the Lebesgue differentiation theorem this limit is well-defined if the measure $\overline \mu_t$ given by $\overline \mu_t(A):=\mu_t( A\setminus\{0\})$  is absolutely continuous with respect to the Lebesgue measure around $0$. Since this is not necessarily satisfied, we can alternatively consider the following definition. 
As the spot price corresponds in reality to a delivery over one of the next hours, 
we can fix some $\tau^*> 0$ and define $S_t$ via $S_t=\frac{1}{\tau^*}  \mu_t((0,\tau^*]).$
\end{remark}

\begin{remark}\label{rem:supp}
Observe that given a $M_+(E)$-valued martingale $(\nu_t)_{t\in[0,T]}$ the $M_+(E)$-valued process $(\mu_t)_{t\in[0,T]}$
obtained via
$$\mu_t(A)=\nu_t(A+t)+\nu_t([0,t))\delta_0(A)$$
satisfies the HJM drift's condition \eqref{eq:martmeas}. This observation is interesting since it permits to conclude that
$$\supp(\mu_t)\subseteq E\cap (E-t)\cup \{0\}=[0,T-t].$$
Since $\supp(\nu_t)\subseteq \supp(\nu_0)$ we can even get that
$$\supp(\mu_t)\subseteq E\cap (\supp(\nu_t)-t)\cup \{0\}=E\cap(\supp(\nu_0)-t)\cup\{0\}.$$

To prove that $(\mu_t)_{t\in[0,T]}$ satisfies \eqref{eq:martmeas}   we proceed as in the proof of Lemma~\ref{lem:mart} computing the dynamics of $\langle \phi,\nu_t\rangle\psi(t)$.
\begin{align*}
\langle \phi,\nu_t\rangle\psi(t)
   &=\langle \phi,\nu_0\rangle\psi(0)+ \int_0^t\langle \phi,\nu_s\rangle\psi'(s)ds+\text{(martingale)}\\
   &=\langle \Psi(0,\cdot), \nu_0\rangle+\int_0^t\langle \frac d {dt}\Psi(s,\fdot),\nu_s\rangle ds+\text{(martingale)},
\end{align*}
for $\Psi(s,x)=\psi(s)\phi(x)$. Extending it to all continuous bounded maps and applying it to $\Psi(s,x)=\Phi((x-s)^+)$ for any $C^1$-function $\Phi$ with $\Phi'(0)=0$ we get
\begin{align*}
\langle \Phi,\mu_t\rangle
&=\langle \Phi(\fdot-t),\nu_t\rangle+\Phi(0)\nu_t([0,t)]\\
&=\langle \Phi((\fdot-t)^+),\nu_t\rangle\\
   &=\int_0^t-\langle \Phi'(\fdot-s)\mathbbm 1_{\{\fdot\geq s\}},\nu_s\rangle ds+\text{(local martingale)}\\
      &=\int_0^t-\langle \Phi',\mu_s\rangle ds+\text{(local martingale)},
\end{align*}
as claimed.
\end{remark}

We explain now how \eqref{eq:martmeas} can be relaxed without compromising the results of Theorem~\ref{th:main1}. We also illustrate that these changes just affect the dynamics of the mass that $\mu_t$ puts in $\{0\}$, which does not play a role for \eqref{eq:future1}. 
\begin{remark}\label{rem:remarkHJMcond}
Suppose now that \eqref{eq:martmeas} is replaced by 
\begin{align}\label{eq:martmeas1}
\langle \phi, \mu \rangle+  \int_0^{\cdot} \Big(\langle \phi', \mu_s \rangle + \gamma\phi(0) \mu_s(0)\Big)
ds \end{align}
being a  $\mathbb{Q}$-local martingale on $[0,T]$ for all $\phi \in D_1$ for some $\gamma\in \R$. 
The  extra term $\gamma\phi(0) \mu_s(0)$ regulates the amount of extra mass that can disappear ($\gamma>0$) or appear ($\gamma<0$) from the set $\{0\}$. 
Observe indeed that  if $(\mu_t)_{t\in[0,T]}$ satisfies \eqref{eq:martmeas} then 
the $M_+(E)$-valued process $$\nu_t:=\mu_t+\gamma_t\delta_0$$ satisfies \eqref{eq:martmeas1} for 
$\gamma_t:=-\gamma e^{-\gamma t}\int_0^t\mu_s(0)e^{\gamma s}ds.$ 

Indeed, by the product formula we  have
$$\gamma_t=\Big(-\gamma e^{-\gamma t}\Big)\Big(\int_0^t\mu_s(\{0\})e^{\gamma s}ds\Big)
=-\gamma\int_0^t\gamma_s+\mu_s(0) ds=-\gamma\int_0^t \nu_s(0) ds,
$$
where the last equality follows since $\gamma_s+\mu_s(0)=\nu_s(0)$.  As  $\langle \phi',\mu_s\rangle=\langle \phi',\nu_s\rangle$ for $\phi\in D_1$, this  implies
\begin{align*}
    \langle \phi,\nu_t\rangle&=\langle \phi,\mu_t\rangle+\phi(0)\gamma_t\\
&=-\int_0^t\langle \phi',\nu_s\rangle +\gamma \phi(0)\nu_s(0) ds+\text{(local martingale)},
\end{align*}
proving the claim.

As for the computation of
\[
F(t, \tau_1, \tau_2)=
\langle \Phi(t,\fdot), \mu_t \rangle,\quad\text{for}\quad \Phi(t,x):=w(t + x; \, \tau_1, \tau_2)\mathbbm{1}_{(\tau_1, \tau_2]}(t+x),
\] 
the mass at  $0$ does not play a role, the assertions of Theorem~\ref{th:main1} hold true if \eqref{eq:martmeas} is replaced by \eqref{eq:martmeas1}. This can be rigorously shown  by adapting the proof of Lemma~\ref{lem:mart}.
 We shall come back to this condition in the next section.
\end{remark}

\begin{remark}\label{rem:discrete}
Note that it is in principle also possible to consider a discrete set of times to maturity, e.g. in the simplest case $E=\{0, 1, 2, \ldots, T\}$ combined with a discrete time setup with $t \in E$.  Then the measure-valued process $\mu$ becomes actually a discrete time stochastic process with values in $\mathbb{R}^{T+1}$ and  $\langle \fdot, \fdot \rangle$ can be identified with the scalar product on $\mathbb{R}^{T+1}$, meaning that $\langle \phi,\mu\rangle=\sum_{i=0}^T\phi(i)\mu(i)$. 

Set $\phi'(i)=(\phi(i)-\phi(i-1))\mathbbm 1_{\{i>0\}}$, assume that $(\mu_t)_{t\in[0,T]}$ is a process in $\R^{T+1}$ with non-negative components, $\E[\sup_{t\in[0,T]}\sum_{i=0}^T\mu_t(i)]<\infty$, and 
\begin{align} \label{eq:martdiscrete}
\langle \phi,
\mu_t\rangle+ \sum_{j=1}^{t} \langle \phi',\mu_{j-1}\rangle
\end{align}
is a local martingale\footnote{Note that, as in the continuous case, with the requested integrability condition local martingality implies martingality.}  on $[0,T]$ for each $\phi:E\to \R$.
Then the (discounted) price of the future 
$$F(t, \tau_1, \tau_2)
=\sum_{ i=0 }^T \mathbbm{1}_{( \tau_1, \tau_2]}(t+i) w(t+i; \tau_1, \tau_2) \mu_t(i)$$
is a martingale on $[0,\tau_1]$ for each $w:\R\times E\times E\to\R$.

This result can be shown following the proof of Theorem~\ref{th:main1}. In particular,  
let $Z$ be the local martingale given by \eqref{eq:martdiscrete} and note that
 \begin{align*}
     \langle \phi,\mu_t\rangle&=\langle \phi,\mu_0\rangle+\sum_{j=1}^t\big(-\langle \phi',\mu_{j-1}\rangle+(Z_j-Z_{j-1})\big),\\
     \psi(t)&=\psi(0)+\sum_{j=1}^t\psi'(j),
 \end{align*}
for each $\phi,\psi:E\to\R$. The summation by parts formula  yields
 \begin{align*}
     \psi(t+1)\langle \phi,\mu_t\rangle
     &=\psi(0)\langle \phi,\mu_0\rangle
     +\sum_{j=0}^t\psi'(j+1)\langle \phi,\mu_j\rangle
     +\sum_{j=1}^t-\psi(j)\langle \phi',\mu_{j-1}\rangle
     \\&\quad+\sum_{j=1}^t\psi(j)(Z_j-Z_{j-1}),
 \end{align*}
 and by reordering the terms and noting that $\sum_{j=1}^t\psi(j)(Z_j-Z_{j-1})$ is a local martingale on $[0,\tau_1]$ we get
  \begin{align*}
     \psi(t)\langle \phi,\mu_t\rangle&
     -\psi(0)\langle \phi,\mu_0\rangle\\
     &=
     \sum_{j=0}^{t-1}\psi'(j+1)\langle \phi,\mu_j\rangle
     +\sum_{j=1}^t-\psi(j)\langle \phi',\mu_{j-1}\rangle
     +(\text{local martingale})\\
     &=
     \sum_{j=1}^t\Big(\psi'(j)\langle \phi,\mu_{j-1}\rangle
     -\psi(j)\langle \phi',\mu_{j-1}\rangle\Big)
  +(\text{local martingale})\\
          &=
     \sum_{j=1}^t\sum_{i=0}^T\Big( \psi'(j)\phi(i)
     -\psi(j)\phi'(i)\Big)\mu_{j-1}(i)
    +(\text{local martingale})\\
     &=
     \sum_{j=1}^t\sum_{i=0}^T\Big( -\psi(j-1)\phi(i)
     +\psi(j)\phi(i-1)\Big)\mu_{j-1}(i)
    +(\text{local martingale}).
 \end{align*}
Proceeding as in Lemma~\ref{lem:mart} we can then show that for every $\Phi: \R \times E\times E\to  \mathbb{R}$ the process
\begin{align*}
\langle \Phi(t,\fdot), \mu_t \rangle+ \sum_{j=1}^t \sum_{i=1}^T \Big(\Phi(j-1,i)-\Phi(j,i-1)\Big) 
\mu_{j-1}(i),
\end{align*}
is a true martingale on $[0,\tau_1]$.
Since setting $\Phi(t,i)=\sum_{ i=0 }^T \mathbbm{1}_{( \tau_1, \tau_2]}(t+i) w(t+i; \tau_1, \tau_2) \mu_t(i)$ it holds
$$\Phi(j-1,i)-\Phi(j,i-1)=0,$$
proceeding as in the proof of Theorem~\ref{th:main1} we can conclude that  $(F(t, \tau_1, \tau_2))_{t\in[0,\tau_1]}$ is a martingale too.

Analogously to \eqref{eq:martmeas1} we can also add the term
$\sum_{j=1}^t \phi(0) \mu_{j-1}(0)$ to \eqref{eq:martdiscrete} and require that
\begin{align}\label{eq:martdiscrete1}
\langle \phi, \mu_t \rangle+ \sum_{j=1}^{t} \Big( \langle\phi', \mu_{j-1}\rangle +\gamma \phi(0) \mu_{j-1}(0) \Big), 
\end{align}
is a martingale on $[0,\tau_1]$. Accordingly for $\Phi: [0,\tau_1] \cap E\times E\to  \mathbb{R}$
the expression 
\[
\langle \Phi(t,\fdot), \mu_t \rangle+ \sum_{j=1}^t \sum_{i=1}^T\Big( (\Phi(j-1,i)-\Phi(j,i-1))  \mu_{j-1}(i)+ \gamma\Phi(j,0) \mu_{j-1}(0)\Big)
\]
is then a martingale. As above this does not change anything for the future prices since $\Phi(j,0)=\mathbbm{1}_{(\tau_1, \tau_2 ]}(j)w(j, \tau_1, \tau_2)=0$ for all $j \in \{1, \ldots t\}$ as $t \in[0, \tau_1]$. 
\end{remark}

\section{Measure-valued diffusions satisfying the HJM condition} \label{sec:measurediff}

So far we implicitly assumed the existence of a \emph{non-negative} measure-valued process satisfying the HJM-drift condition \eqref{eq:martmeas}.
 To analyze when this non-negativity  together with \eqref{eq:martmeas} holds true, we consider a Markovian setting where the process $(\mu_t)_{t\in [0,T]}$ can be described by its (extended) generator. To this end we need to introduce spaces of functions on which these operators act. These are  cylindrical functions and polynomials with certain regularity properties.
 
\subsection{Cylindrical  functions and polynomials}
\label{sec:cylind}

Recall that 
$
D \subseteq C(E)
$
is a dense linear subspace containing the constant function 1.
The set of \emph{cylindrical functions} we  consider are  maps from $M(E)$ to $\mathbb{R}$ lying in
$$
F^{D}=\Big\{ \nu \mapsto f(\nu)=\Psi(\langle g_1, \nu \rangle, \ldots, \langle g_m, \nu \rangle)
\colon   \Psi \in C^{\infty}(\mathbb{R}^m),  g_k \in D,\,  m\in \mathbb{N}_0\Big\}.
$$
We shall denote the restriction of $F^D$ to $M_+(E)$ by 
\[
F^D(M_+(E)):=\{f|_{M_+(E)}\colon f\in F^D\}.
\]
Moreover, since elements in $F^D(M_+(E))$ are not necessarily compactly supported we  also define the following sets
\begin{align}\label{F^D_c}
F^D_0:=F^D(M_+(E)) \cap C_0(M_+(E))\qquad\text{and}\qquad
F^D_c:=F^D(M_+(E)) \cap C_c(M_+(E)).
\end{align}
These sets will serve as domain for the linear operators that we introduce subsequently.

Similarly to cylindrical functions, we  also consider cylindrical polynomials. To this end
denote by $\Pol(\R^m)$ the set of polynomials on $\R^m$. We then call  a map $p: M(E) \to \mathbb{R}$ \emph{cylindrical polynomial} if it lies in the set
$$
P^{D}:=\Big\{\nu \mapsto p(\nu)=\Psi(\langle g_1, \nu \rangle, \ldots, \langle g_m, \nu \rangle)
\colon   \Psi \in \Pol(\R^m),\, g_k \in D,\,  m\in \mathbb{N}_0\Big\}.
$$
Note that $P^D$ consists of all (finite) linear combinations of the constant polynomial and ``rank-one'' monomials $\langle g\otimes\cdots\otimes g,\nu^k\rangle=\langle g,\nu\rangle^k$ with $g\in D$.
Analogously to cylindrical functions we write $P^D(M_+(E))$ for the restriction of $P^D$ to $M_+(E)$.

\subsection{Directional derivatives}

In order to be able to consider certain differential operators acting on functions with measure arguments, we need to introduce the notion of derivatives.
We shall throughout apply 
directional derivatives. More precisely, 
a function $f\colon M(E)\to\R$ is called differentiable at $\nu$ in direction $\delta_x$ for $x\in E$ if
\[
\partial_x f(\nu) := \lim_{\varepsilon\to0} \frac{f(\nu+\varepsilon\delta_x)-f(\nu)}{\varepsilon}
\]
exists. We write $\partial f(\nu)$ for the map $x\mapsto\partial_xf(\nu)$, and we use the notation
\[
\partial^k_{x_1x_2\cdots x_k} f(\nu) := \partial_{x_1}\partial_{x_2}\cdots \partial_{x_k} f(\nu)
\]
for iterated derivatives. We write $\partial^k f(\nu)$ for the corresponding map from $E^k$ to $\R$.

\begin{remark}
Note for all functions in $F^D$ and $k \in \mathbb{N}$, it holds $\partial^k f(\nu) \in D^{\otimes k}$. In particular, setting $f(\nu)=\Psi(\langle g_1, \nu \rangle, \ldots, \langle g_m, \nu \rangle)$  and $\vec g:=(g_1,\ldots,g_m)^\top$ we get
\begin{align*}
    \partial_{x}f(\nu)&=\nabla \Psi(\langle \vec g, \nu \rangle)^\top\vec g(x),\\
\partial^2_{xy}f(\nu)&=\vec g(x)^\top\nabla^2 \Psi(\langle \vec g, \nu \rangle)\vec g(y).
\end{align*}
\end{remark}

\subsection{Diffusion-type operators and martingale problems}

We shall consider here diffusion type operators that correspond to second-order differential operators, precisely introduced in the next definition. As domain we consider a generic set $D(L)\subseteq F^D$. Later on $D(L)$ will be given by $F^D_c$ as defined in \eqref{F^D_c} but also by other sets which are more convenient in the polynomial setting (see Section~\ref{sec:poly}).

\begin{definition}
A linear operator $L:D(L)\to C_0(M_+(E))$ is called diffusion-type operator if  it admits a representation
\begin{align}\label{eq:diff}
Lf(\nu) = &B(\partial f(\nu), \nu) + \frac{1}{2} Q(\partial^2 f(\nu), \nu)
\end{align}
for some operators $B:D \times M_+(E) \to\R$ and $Q:D\otimes D \times M_+(E) \to \mathbb{R}$ such that $B(\fdot, \nu)$ and $Q(\fdot, \nu)$ are linear for all $\nu \in M_+(E)$.
\end{definition}

To such operators we can then associate measure-valued processes via the martingale problem.

Let $L$ be a linear operator acting on $D(L)\subseteq F^D$  and being of form \eqref{eq:diff}.
A $M_+(E)$-valued process $(\mu_t)_{t\in[0,T]}$ with continuous trajectories defined on some filtered probability space $(\Omega, \Fcal,(\Fcal_t)_{t\in [0,T]}, \mathbb{Q})$ \footnote{We here consider directly a risk neutral measure $\mathbb{Q}$ because we are interested when \eqref{eq:martmeas} is a martingale under $\mathbb{Q}$.} is called a $M_+(E)$-valued {\em solution to the martingale problem for $L$} with initial condition $\nu\in M_+(E)$ if $X_0=\nu$ $\mathbb{Q}$-a.s.~and
\begin{equation}\label{martprob}
N^f_t = f(\mu_t) - f(\mu_0) - \int_0^t Lf(\mu_s) ds, \quad t \in [0,T]
\end{equation}
defines a local martingale for every $f$ in the domain of $L$.
Uniqueness of solutions to the martingale problem is always understood in the sense of law. The martingale problem for $L$ is {\em well--posed} if for every initial condition $\nu \in M_+(E)$ there exists a unique $M_+(E)$-valued solution to the martingale problem for $L$ with initial condition~$\nu$.

\subsection{Existence of measure-valued diffusions satisfying the HJM drift-condition}

We are now concerned with proving existence of $M_+(E)$-valued-solution to such types of martingale problems. The following theorem provides sufficient conditions. For its formulation recall the space of  $F^D_c$ as defined in \eqref{F^D_c} and  the definition of $D_1$  given in \eqref{eq:D1}.

\begin{theorem}\label{mainexistence}
Let  $L:F^{D_1}_c\to C_0(M_+(E))$ be  a linear operator of form \eqref{eq:diff} for some $Q$ such that 
$ \nu \mapsto Q(\partial^2 f(\nu), \nu) \in C_0(M_+(E))$ for all $f \in F^{D_1}_c$. 
Suppose that the drift part $B$ is given by
\begin{align}\label{eq:B}
B(\partial f(\nu), \nu)= - \langle \frac{d}{dx} \partial f(\nu), \nu \rangle.
\end{align}
If the diffusion part $Q$ satisfies the positive maximum principle on $M_+(E)$, i.e.~
\[
f \in F^{D_1}_c, \, \nu^* \in M_+(E), \sup_{M_+(E)}f = f(\nu^*) \geq 0 \quad \text{implies} \quad Q(\partial^2 f(\nu^*), \nu^*) \leq 0,
\]
then there exists a $M_+(E)$-valued solution to the martingale problem for $L$ which satisfies the HJM-drift condition \eqref{eq:martmeas}.
\end{theorem}

\begin{proof}
We verify now that the conditions of Lemma~\ref{lem:martingaleprob} are satisfied. 
First note that $D_1$ is a dense subset of $C(E)$ which contains $1$ and
$$ B(\partial f(\fdot), \fdot)= - \langle \frac{d}{dx} \partial f(\fdot), \fdot \rangle$$
lies in $ C_0(M_+(E))$
for all $f \in F^{D_1}_c$. Moreover, note that
\[
B( 1, \nu)  =0,
\]
whence \eqref{eq:growth} is satisfied. Since $Q$ satisfies the positive maximum principle by assumption, it remains to prove that $B$ does it as well. By \cite[Theorem~3.1]{CGPS:21} we know that for each $f,\nu^*$ such that
$$f \in F^{D_1}_c, \, \nu^* \in M_+(E), \sup_{M_+(E)}f = f(\nu^*) \geq 0$$
it holds
$\partial_x f(\nu^*) \leq 0$ for each $x\in E$ with equality for $x\in\supp(\nu^*)$. Since $E=[0,T]$ this in particular implies that $\partial_x f(\nu^*)$ has a maximum over $[0,T]$ in $x$ for each $x\in \supp(\nu^*)$ and thus $-\frac{d}{dx} \partial_x f(\nu^*) = 0$ for each $x\in \supp(\nu^*)\cap(0,T)$ and $-\frac{d}{dx} \partial_T f(\nu^*) \leq 0$. Since $\partial f(\nu^*)\in D_1$ we also know that its derivative vanishes in 0 and thus
\[
B(\partial f(\nu^*),\nu^*)=\langle -\frac{d}{dx} \partial f(\nu^*), \nu^* \rangle \leq 0,
\]
what we needed to prove. 

\end{proof}

\begin{remark}
If condition \eqref{eq:martmeas} is replaced by \eqref{eq:martmeas1} as outlined in Remark
\ref{rem:remarkHJMcond}, then condition \eqref{eq:B} translates to 
 \begin{align}\label{eq:Balt}
 B(\partial f(\nu), \nu)=-\langle \frac{d}{dx} \partial f(\nu) + \gamma\mathbbm{1}_{\{x=0\}}  \partial f(\nu), \nu \rangle.
 \end{align}
Note however that $ \nu \mapsto B(\partial f(\nu), \nu)$ is not continuous in general, since $\nu \mapsto \nu(0)$ is  not continuous. This could be recovered by restricting $D_1$ to maps vanishing in 0, but this new set is not dense in $C(E)$. To overcome the technical problems introduced by the lack of density we would need to slightly adapt all our technical results (in particular Lemma~\ref{lem:martingaleprob}) replacing $C(E)$ with $C_0(E\setminus\{0\})$. For seek of exposition we will not deepen this direction any further.

Note that for $\gamma=0$ the drift operator $B$
is the generator of the probability measure-valued process which describes the ``law'' of the deterministic spatial motion $dX_t=-dt$, 
absorbed at $0$ due the definition of $D_1$. Adding $-\langle \gamma\mathbbm{1}_{\{x =0\}}\partial f(\nu), \nu \rangle$ to $B$, which amounts to consider the operator given in \eqref{eq:Balt}, can be interpreted as adding a killing term to the spatial motion, meaning that the particles are killed at rate $\gamma$ after being absorbed at $0$. Indeed, in this case the generator of the spatial motion is given by
\[
\mathcal{A}g= -\frac{d}{dx}g - \gamma \mathbbm{1}_{\{x =0\}}g.
\]
\end{remark}

\begin{remark}
Note that in order to apply Theorem~\ref{th:main1}, the solution to the martingale problem $(\mu_t)_{t \in [0,T]}$ provided by Theorem~\ref{mainexistence} additionally needs to satisfy \eqref{eq:finitesup}, namely
$$
\mathbb{E}_{\mathbb{Q}}[\sup_{t \in [0,T]} \mu_t(E)] < \infty.
$$
By the HJM-drift condition  we know that $ \mu(E)= \langle 1, \mu \rangle$ is a local martingale. The BDG inequalities then yield
\[
\mathbb{E}_{\mathbb{Q}}[\sup_{t \in [0,T]} \mu_t(E)] \leq C  \mathbb{E}_{\mathbb{Q}}[ [\mu(E)]_T^{1/2}].
\]
If $\mu(E)$ is a true martingale, then we can also apply Doob's $L_1$-martingale inequality to get
\[
\mathbb{E}_{\mathbb{Q}}[\sup_{t \in [0,T]} \mu_t(E)] \leq \frac{e}{e-1} (1+ \mathbb{E}_{\mathbb{Q}}[ \mu_T(E) \log \mu_T(E)] ).
\]
\end{remark}

\begin{remark}
In the discrete time setup of Remark \ref{rem:discrete} choosing $\phi(x)=1_{\{x=i\}}$ we can see that for each $t>0$ condition \eqref{eq:martdiscrete} yields
$$\mu_t(i)-\mu_{t-1}(i)=-\langle \phi',\mu_{t-1}\rangle=-(\mu_{t-1}(i)1_{\{i>0\}}-\mu_{t-1}(i+1)).$$
As a result 
the  $\mathbb{R}^{T+1}$ valued increments of $\mu$ can be described via
\[
\mu_t- \mu_{t-1} = \beta\mu_{t-1} +\Delta M_t, 
\]
where $\Delta M_t$ denotes a martingale increment and $\beta$ is given by
\[
\beta=\begin{pmatrix}
0 & 1 & 0 &\ldots& \ldots& 0 \\
0 & -1 & 1 &0 &\ldots& 0\\
\vdots& & \vdots & &\vdots&\\
\vdots& & \vdots & &\vdots&\\
0 & \ldots &&\ldots&-1 & 1\\
0 & \ldots &&\ldots & 0 & -1
\end{pmatrix}.
\]
The process corresponding to \eqref{eq:martdiscrete1} for $\gamma>0$ is similar only with $\beta_{11}=-\gamma$ instead of $0$. In both cases the  matrix $\beta^{\top}$ corresponds to the generator matrix of the spatial motion. In the first case it describes (deterministic) downward jumps by $1$ being absorbed at $0$, while in the second case the absorption at $0$ is followed by a killing with rate $\gamma$. 
\end{remark}

\section{Tractable examples of affine and polynomial type}\label{sec:poly}

This section is dedicated to tractable examples of polynomial and affine type that satisfy the conditions of Theorem~\ref{mainexistence}.
For these specifications the moments as well as the Laplace transform (in the affine case) can be computed explicitly by solving a system of linear PDEs or Riccati PDEs respectively. Throughout we shall  implicitly work under $\mathbb{Q}$, but do not always indicate this in the expected values.
 
We start by providing the definition and a characterization result of polynomial operators given in \cite{CGPS:21}. Their domain is always given by
\begin{align}\label{eq:domain}
\mathcal{D}:=\operatorname{span}\{ P^D(M_+(E)), F^D_0  \}.
\end{align}
instead of $F^D_c$ as in previous sections.

To define polynomial operators we need to introduce the notion of polynomials on $M(E)$. A polynomial of degree $m$ on $M(E)$ is given by i.e.
$$
p(\nu) = \sum_{k=0}^m \langle g_k,\nu^k\rangle,
$$
for some coefficients $g_k\in\widehat C(E^k)$ with $g_m\neq 0$, where
\[
\langle g, \nu^k \rangle := \int_{E^k} g(x_1,\ldots,x_k) \nu(dx_1) \cdots \nu(dx_k).
\]
 The set of all polynomials is denoted by $P$. Recall also the notion of cylindrical polynomials $P^D$ as introduced in Section \ref{sec:cylind}. Note in particular that setting $g_k:=g^{\otimes k}$ we get $\langle g,\nu\rangle^k=\langle g_k,\nu^k\rangle$, showing that $P^D\subseteq P$.

\begin{definition}
A linear operator  $L:\Dcal\to C(M_+(E))$ is called $M_+(E)$-polynomial if
it maps $P^D$ to  $P$  such that for every $p\in P^D$ there is some $q\in P$ whose degree is smaller or equal the degree of $p$ and $q|_{M_+(E)}=Lp|_{M_+(E)}$.
\end{definition}

The form of these operators in a diffusion setup can be characterized completely. For the proof of a corresponding statement see Theorem~4.6 in \cite{CGPS:21}.
\begin{theorem}
A linear operator $L:\mathcal{D}\to C(M_+(E))$ is $M_+(E)$-polynomial and of diffusion-type if and only if $L$ admits a representation
\begin{align}
Lf(\nu) = &B_0(\partial f(\nu)) + \langle B_1(\partial f(\nu)),\nu\rangle\label{eq:op}\\
& + \frac{1}{2} \Big(Q_0(\partial^2 f(\nu)) + \langle Q_1(\partial^2 f(\nu)), \nu\rangle + \langle Q_2(\partial^2 f(\nu)), \nu^2\rangle\Big), \quad f\in \Dcal,  \nu\in M_+(E)\notag
\end{align}
for some linear operators $B_0:D\to\R$, $B_1:D\to  C(E)$, $Q_0:D\otimes D\to\R$, $Q_1:D\otimes D\to C(E)$, $Q_2:D\otimes D\to \widehat C(E^2)$. These operators are uniquely determined by $L$.
\end{theorem}

We now combine Theorem~\ref{mainexistence} with one of the main results of \cite{CGPS:21}, which gives  sufficient conditions for the existence of solutions to the martingale problem for  polynomial operators.  To state the result we first need the following definition.

\begin{definition}
We say that a linear operator $C$ admits a  $(\beta, \pi)$-representation if
$$
C(g)(x,y)=\frac{1}{2}(\pi(x,y) g(x,x) + \pi(y,x) g(y,y) +2 \beta(x,y) g(x,y)), \quad  g \in D \otimes D
$$
where
\begin{itemize}
\item  $\beta: E^2 \to \mathbb{R}$ is a symmetric function such that $\beta(x,x) \geq 0$ for all $x \in E$;
\item $\pi: E^2 \to \mathbb{R}_+$ is a non-negative function  such that $\pi(x,x)=0$ for all $x \in E$;
\item for all $n\in \mathbb{N}$, $x_1, \ldots, x_n \in E$, $c_1, \ldots, c_n \in \mathbb{R}_{++}$, the $n \times n$ matrix 
\begin{align}\label{eq:matrix}
A^{(n)}:=\beta_n+ \begin{pmatrix} \sum_{j=1}^n \frac{c_j}{c_1} \pi(x_1, x_j) &  &\\
& \ddots &\\
& & \sum_{j=1}^n \frac{c_j}{c_n} \pi(x_n, x_j)
 \end{pmatrix} \in \mathbb{S}^n_+,
\end{align}
where $\beta_n \in \mathbb{S}^n$ with entries $\beta_{n,ij}=\beta(x_i,x_j)$;
\item the map $ (x,y) \mapsto \frac{1}{2}(\pi(x,y) g(x,x) + \pi(y,x) g(y,y)+ 2\beta(x,y)g(x,y))$ lies in $\widehat{C}(E^2)$ for all $g \in D\otimes D$. 

\end{itemize}
\end{definition}

Accordingly, sufficient conditions for the  existence of the martingale problem such that additionally the HJM-condition \eqref{eq:martmeas} is satisfied  can now be formulated as follows. Note that here the local martingale condition \eqref{martprob} holds for all $f \in \mathcal{D}$ with $\mathcal{D}$ defined in \eqref{eq:domain}.

\begin{theorem}\label{main1}
Let $D=D_1$ and $L:\Dcal\to C(M_+(E))$\footnote{Note that $D_1$ enters in the definition of $\mathcal{D}$ given by \eqref{eq:domain}.} be a linear operator of form~\eqref{eq:op}, where
\begin{enumerate}
\item $B_0\equiv 0$;
\item $B_1=-\frac{d}{dx}$;
\item $Q_0 \equiv 0$;
\item $Q_1$ is of the form 
\[
Q_1(g)= \alpha \diag(g), \quad \text{ that is,} \quad  Q_1(g)(x)=\alpha(x)g(x,x), \quad g \in D \otimes D,
\]
where $\alpha \in C(E)$ with values in $\mathbb{R}_+$;
\item $Q_2$ 
admits a ($\beta, \pi$)-representation.
\end{enumerate}
Then $L$ is $M_+(E)$-polynomial and for every initial condition $\nu\in M_+(E)$ its martingale problem has a $M_+(E)$-valued solution  that has continuous paths and satisfies the HJM-drift condition \eqref{eq:martmeas}. Moreover, condition \eqref{eq:finitesup} holds.
\end{theorem}
The immediate consequence of this theorem and Theorem~\ref{th:main1} is given in the following corollary.
\begin{corollary}
Let $L$ be as in Theorem~\ref{main1} and $(\mu_t)_{t\in[0,T]}$ be a solution of the corresponding martingale problem. Then the corresponding market satisfies NAFLVR in the sense of Assumption~\ref{ass:NAFLVR}. 
\end{corollary}

\begin{proof}[Proof of Theorem~\ref{main1}]
 The existence's result follows by Theorem~5.9 in \cite{CGPS:21}, of which we check now the conditions. Observe that for each $\phi\in D_1$, $x\in E$ such that $\phi(x)=\sup_E\phi=0$ it holds
$$B_1\phi(x)=0,\ x\in(0,T),\qquad B_1\phi(T)\leq0,\qquad B_1\phi(0)=0,$$
where the last property follows by definition of $D_1$. This implies that $B_1\phi(x)\leq0$, hence that $B_1$ satisfies the positive minimum principle on $E$, and by Theorem~5.9 in \cite{CGPS:21} we can conclude that $L:\Dcal\to C(M_+(E))$ is $M_+(E)$-polynomial and for every initial condition $\nu\in M_+(E)$ its martingale problem has a $M_+(E)$-valued solution. Since for $f(\nu):=\langle 1,\nu\rangle$ we have $f\in \Dcal$ and $\partial f\equiv 1$,  the HJM-drift condition \eqref{eq:martmeas} follows by the martingale problem.

Finally, by \cite[Lemma~5.7]{CGPS:21} we have $\mathbb{E}[\langle 1,\mu_t\rangle^2] < \infty$ and $\langle 1,\mu\rangle$ is a true martingale. By Doob's maximal inequality condition \eqref{eq:finitesup} follows.
\end{proof}

As a next step, we are now interested in the so-called moment formula to get explicit expressions for the moments. Before stating it we need to associate to an $M_+(E)$-polynomial operator $L$ a family of so-called \emph{dual} operators  $(L_k)_{k \in \mathbb{N}}$. These operators are linear operators mapping the coefficients vector $\vec g:=(g_0, \ldots, g_m) \in \bigoplus_{k=0}^m \in D^{\otimes k}$ of a polynomial of the form $p= \sum_{k=0}^m \langle g_k, \nu^k \rangle$ to the coefficients vector of $Lp$. 
To introduce the dual operator, we now recall Definition~2.3 of \cite{CS:19}.

\begin{definition}
Fix  $m\in \N_0$ and let $L$ be an $M_+(E)$-polynomial operator. 
The \emph{$m$-th dual operator} corresponding to $L$ is a linear operator  $L_m: \bigoplus_{k=0}^m D^{\otimes k}\to
\bigoplus_{k=0}^m \widehat{C} (E^k)$
such that $L_m\vec g =:(L_m^0\vec{g},\ldots, L_m^m\vec{g})$ satisfies
$$
Lp(\nu)= \sum_{k=0}^m \langle L^k_m \vec{g},\nu^k\rangle
\qquad \text{for all } \nu \in M_+(E),
$$
where $p(\nu) = \sum_{k=0}^m \langle g_k,\nu^k\rangle$.
 Whenever $L_m$ is a closable operator\footnote{We refer to \cite[Chapter 1]{EK:09} for the precise definition.}, we still denote its closure  by $L_m: \mathcal{D}(L_m)\to \bigoplus_{k=0}^m \widehat{C} (E^k)$ and its domain by $\mathcal{D}(L_m)\subseteq \bigoplus_{k=0}^m \widehat{C} (E^k)$. 
\end{definition}

\begin{remark}\label{rem:dual}
If $L$ satisfies the conditions of Theorem~\ref{main1}, then the $m$-th dual operator $L_m$ 
when applied to $\vec{g}=(0, \ldots, 0, g^{\otimes m})$
is of the form
\begin{equation}\label{eq:Lm}
\begin{split}
L_m^m(g^{\otimes m})&=
-m g'\otimes g^{\otimes(m-1)}+\frac {m(m-1)}2 Q_2(g\otimes g)\otimes g^{\otimes (m-2)},\\
L_m^{m-1}(g^{\otimes m})&=\frac {m(m-1)}2 Q_1(g\otimes g)\otimes g^{\otimes (m-2)},\\
L_m^k(g^{\otimes m})&=0 \quad k \leq m-2,
\end{split}
\end{equation} 
Identifying $\vec{g}=(0, \ldots, 0, g^{\otimes m})$ with $g^{\otimes m}$, we write here
by a slight abuse of notation  $L_m^k(g^{\otimes m})$ instead of  $L_m^k\vec{g}$. Note that $L^k_m$ applied to general $\vec{g}$ is not $0$ for $k \leq m-2$. Note also that  $L^0_m\equiv 0$.
\end{remark}

Let $L$ satisfy the conditions of Theorem~\ref{main1}. Fix  $m\in \N_0$, and let $L_m$ be the closable $m$-th dual operator corresponding to $L$ with domain $\mathcal{D}(L_m)$.
Before stating the moment formula in the current context,  we extend the domain of $L$ to polynomials with coefficients in $\mathcal{D}(L_m)$ by setting
$$
p_{\vec{g}}=\sum_{k=0}^m \langle g_k,\nu^k\rangle \quad \text{and} \quad Lp_{\vec{g}}=\sum_{k=1}^m \langle L^{k}_m\vec{g},\nu^{k}\rangle,
$$
for all $\vec{g} \in \mathcal{D}(L_m)$ and $m\in \N_0$. We here used that $L^0_m\equiv 0$.
As stated below, the moment formula corresponds to a solution of a system of linear PDEs. Note  that in contrast to \cite{CGPS:21} we shall call these differential equations PDEs since the underlying space is here always $E=[0,T]$.  In the next definition we recall the precise solution concept.

\begin{definition}\label{def:sol}
Let $L_m$ be given by  \eqref{eq:Lm}.
We call  a function $t \mapsto \vec{g}_t$ 
with values in $\mathcal{D}(L_m)$  a solution of the $m+1$ dimensional system of PDEs
\[
\frac d {dt} \vec{g}_t =  L_m \vec{g}_t , \quad \vec{g}_0=(g_{0,0}, g_{0,1}, \ldots, g_{0,m})		
\]
if for every $t >0$ it holds
\begin{align}\label{eq:ODE}
\sum_{k=0}^m \langle g_{t,k},\nu^k\rangle=\sum_{k=0}^m \langle g_{0,k},\nu^k\rangle +\int_0^t  \sum_{k=1}^m \langle L^{k}_m \vec{g}_s, \nu^{k}  \rangle  ds
\end{align}
for all $\nu \in M_+(E)$.
\end{definition}

\begin{remark}
Note that this solution concept reduces to a more classical one if we take $\nu=\delta_{x_1} + \cdots +\delta_{x_k}$ with $x_i \in E$, $i=1, \ldots, k$ and $k =1, \ldots, m$. Indeed, by polarization \eqref{eq:ODE} can be transformed into
\begin{align*}
g_{t,0}&=g_{0,0}\\
g_{t,1}(x_1)&=g_{0,1}(x_1)+ \int_0^t L^1_m \vec{g}_s(x_1) ds\\
&\vdots\\
g_{t,m-1}(x_1, \ldots, x_{m-1})&=  g_{0,m-1}(x_1, \ldots, x_{m-1}) +\int_0^t L^m_{m-1} \vec{g}_s(x_1, \ldots, x_{m-1}) ds\\
g_{t,m}(x_1, \ldots, x_m)&=  g_{0,m}(x_1, \ldots, x_m) +\int_0^t L^m_m \vec{g}_s(x_1, \ldots, x_m) ds
\end{align*}
and thus reduces to a classical (except of the integral form) solution of a multivariate PDE. 
\end{remark}

\begin{theorem}[Dual moment formula]
Let $L$ satisfy the conditions of Theorem~\ref{main1}. Fix  $m\in \N_0$, let $L_m$ be the  $m$-th dual operator corresponding to $L$ as of Remark \ref{rem:dual}, and assume that $L_m$ is closable with domain $\mathcal{D}(L_m)$. 
Fix a coefficients vector $\vec{g}=(g_0, \ldots, g_m)\in\mathcal{D}(L_m)$ and  suppose that the following condition holds true.
\begin{itemize}
\item 
There is a solution  in the sense of Definition~\ref{def:sol} of the $m+1$ dimensional system  of linear ODEs on $[0,T]$ given by
\begin{equation}\label{ODE}
\frac d {dt} \vec{g}_t =  L_m \vec{g}_t, \qquad 
\vec{g}_0=(g_{0,0}, g_{0,1}, \ldots, g_{0,m}).		\end{equation}
\end{itemize}
Then, for any solution $(\mu_t)_{t \in [0,T]}$ to the martingale problem for $L$ and all $0 \leq s \leq t\leq T$ the representation
\begin{equation}\label{eq:momentformula}
\mathbb{E}\left[  \sum_{k=0}^m \langle g_{0,k}, \mu_t^k \rangle \,\bigg|\, \mathcal{F}_s \right ]= \sum_{k=0}^m \langle g_{t-s,k}, \mu_s^k \rangle
\end{equation}
 holds almost surely. 
\end{theorem}

\begin{remark}\label{rem:inhomogen}
Let us remark that the current setup can also be generalized to \emph{time-inhomogenous measure-valued polynomial processes} which have been considered in finite dimensions in \cite{HS:20} (see also \cite{F:05} for the affine case). This then allows in particular to incorporate seasonalities in a tractable manner.
\end{remark}

\begin{example}
Since $L_1^1g=-g'$ and $L_1^0=L_0^0=0$ the corresponding system of PDEs reads
\[
\frac d {dt} \binom{g_{t,0}}{g_{t,1}} =  \binom{0}{-g'_{t,1}}  , \quad \vec{g}_0=(g_{0,0}, g_{0,1}).
\]
Choosing $g_{0,0}=0$ and $g_{0,1}=\phi$, we see that the corresponding solution is given by $\vec g_t=(0,\phi(\fdot -t))^\top$. From the moment formula we can thus conclude that
$$
\mathbb{E}\left[  \langle \phi, \mu_t \rangle \,\big|\, \mathcal{F}_s \right ]=  \langle \phi(\fdot-(t-s)), \mu_s \rangle=\langle \phi, \nu_{t,s} \rangle,
$$
where $\nu_{t,s}([a,b]):=\mu_s([a+t-s,b+t-s])$. Setting $t=T$ and $s=0$, we obtain $\nu_{T,0}$ is supported on $\{0\}$. This implies that $
\mathbb{E}\left[  \langle \phi, \mu_T \rangle  \right ]=0$ for each $\phi$ such that $\phi(0)=0$ and thus that $\mu_T$ is almost surely supported on $\{0\}$, which is also in line with Remark \ref{rem:supp}.

Note, that since we only consider $L_1$, $Q$ does not play a role for this result. This in particular implies that this representations holds for any process $(\mu_t)_{t\in[0,T]}$ satisfying the assumption of Theorem~\ref{th:main1}.
\end{example}

In the following sections we shall give important examples when the system  of linear ODEs  \eqref{ODE} has a solution in the sense of Definition~\ref{def:sol}. Recall that in this case the moment formula can be applied.

\subsection{Black-Scholes type measure-valued HJM-models}
We start by considering the case where the diffusion part corresponds to an analog of a multivariate Black-Scholes model.
\begin{definition}\label{def:BStype}
A \emph{Black-Scholes type measure-valued  HJM-model} is the solution to the martingale problem for some $L:\Dcal \to  C(M_+(E))$ for some $L$ admitting the representation
$$
Lf(\nu) =  \langle -\frac d {dx}(\partial f(\nu)),\nu\rangle
 + \frac{1}{2}  \langle Q_2(\partial^2 f(\nu)), \nu^2\rangle,
$$
for each $f\in \Dcal$,  $\nu\in M_+(E)$,
 where $Q_2:D\otimes D\to \widehat C(E^2)$ admits a ($\beta, \pi$)-representation for some $\beta,\pi$ bounded and continuous on $E^2 \setminus \{x = y\}$.
\end{definition}
Observe that the operator $L$ corresponding to a Black-Scholes type measure-valued  HJM-model satisfies the conditions of Theorem~\ref{main1} for $Q_1 \equiv 0$ and $Q_2$ as in the definition.

\begin{proposition}\label{prop:BStype}
Let $(\mu_t)_{t \in [0,T]}$ be a Black-Scholes type measure-valued HJM-model. Then the moment formula \eqref{eq:momentformula} holds true. 
\end{proposition}

\begin{proof}

Let $L_m$ be the  $m$-th dual operator corresponding to $L$.
Note that when applying it to $\vec{g}=(0, \ldots, 0, g^{\otimes m})$ with $g \in D_1$, we have
\begin{equation}\label{eq:Lm1}
\begin{split}
L^m_m(g^{\otimes m})&=-m g' \otimes g^{\otimes m-1} + \frac{m (m-1)}{2} Q_2(g \otimes g) \otimes g^{\otimes m-2},\\
L^k_m(g^{\otimes m})&=0, \quad k \leq m-1,
\end{split}
\end{equation}
where we write here 
by a slight abuse of notation, similarly as in Remark \ref{rem:dual},  $L_m^k(g^{\otimes m})$ instead of  $L_m^k\vec{g}$.

Observe now that $B^m(g^{\otimes m}):=-m g' \otimes g^{\otimes m-1}=-1^\top \nabla(g^{\otimes m})$ generates the strongly continuous contraction shift semigroup $(P^m_t)_{t \in [0,T]}$ on $\widehat{C}(E^m)$  given by 
\begin{align*}
P^m_tg(x)=
g\big((x_1-t)^+,\ldots,(x_m-t)^+\big), 
\end{align*}
for each $g\in \widehat C(E^m)$.
Moreover, by the form of $Q_2$ the closure of the operator $Q^m(g^{\otimes}) :=  \frac{m (m-1)}{2} Q_2(g \otimes g) \otimes g^{\otimes m-2}$ is a bounded operator on $\widehat{C}(E^m)$. By \cite[Corollary 1.7.2]{EK:09}, $\overline{L}^m_m$ therefore generates a strongly continuous semigroup on $\widehat{C}(E^m)$ which is additionally positive since $Q^m$ (and also $B^m$) satisfies the positive minimum principle (see the end of Section \ref{sec:notation}). Denote  now, for every $k \leq m$, the strongly continuous positive semigroup corresponding to $\overline{L}^k_k$ by
$(Y^{k}_t)_{t\in [0,T]}$ which acts on $\widehat{C}(E^k)$. Then, by the definition of this semigroup 
 \[
 \vec g_t:=(\underbrace{Y_t^0}_{\text{Id}}g_{0,0}, \ldots,Y_t^mg_{0,m}) 
 \] 
 satisfies \eqref{ODE} for any initial value $\vec{g}_0=(g_{0,0}, g_{0,1}, \ldots, g_{0,m}) \in  \mathcal{D}(L_m)$. Hence, the moment formula \eqref{eq:momentformula} holds true.
\end{proof}

\begin{remark}
Observe  that for each $g\in D_1^{\otimes m}$ it holds
$$
    L^m_mg(x)=\Gcal_mg(x)+\gamma(x)g(x),
$$
where $\gamma(x)=\sum_{i=1}^m\sum_{j=1}^m\big(\pi(x_i,x_j)+\beta(x_i,x_j)\big)$ and
$$\Gcal_mg(x):=-1^\top\nabla g(x) + \frac{1}{2} \sum_{i=1}^m\sum_{j=1}^m \pi(x_i,x_j)\big(g(x+\delta_{ij}(x))-g(x)\big)$$
for $\delta_{ij}(x)=(x_i-x_j)e_j$. Observe that $\Gcal_m$ is the generator of the $\R^m$-valued process $Z$ that can be described as follows. Each component of $Z$ can jump to the position of one of the other components. The stochastic intensity for a jump of the $j$-th component to the position of the $i$-th one is given by $\big(\frac 1 2\pi(Z^{i}_t,Z^{j}_t)\big)_{t\in[0,T]}$. Between two jumps each component of $Z$ decreases linearly with slope -1. Since the derivative of each test function vanishes at 0 we can require the process to be absorbed in 0. This means that given two consecutive jump times $\tau_1,\tau_2$ it holds
$$Z^{i}_{\tau_1+t}=(Z^{i}_{\tau_1}-t)^+,\qquad t<\tau_2-\tau_1.$$

Following a Feynman-Kac type of argument one can then show that the PDE  
$$\frac d {dt}\langle g_t,\nu\rangle=\langle L_m^m g_t,\nu\rangle,\qquad g_0=g,$$
corresponding to $L_m^m$ is solved by
$$\langle g_t,\nu\rangle:=\E[e^{-\int_0^t\gamma(Z_s)ds}g(Z_t)|Z_0\sim \nu].$$
By the moment formula we can thus conclude that
$$
\mathbb{E}\left[  \langle g, \mu_t^m \rangle \,\big|\, \mathcal{F}_s \right ]
= \langle g_{t-s}, \mu_s^m \rangle
=\E[e^{-\int_0^{t-s}\gamma(Z_u)du}g(Z_{t-s})|Z_0^1,\ldots Z_0^m \text{ iid},\  Z_0^i\sim \mu_s],
$$
for each $g\in D_1^{\otimes m}$.
\end{remark}

\begin{remark}
The choice $\pi \equiv 0$ in the $(\beta, \pi)$-representation of $Q_2$ can be seen as measure-valued analog of the quadratic variation structure of a multivariate Black-Scholes model. This  was also the inspiration for the name of  the whole model class introduced in Definition~\ref{def:BStype}. Observe that the form of $L$ then simplifies to
\begin{align}\label{eq:L-BS}
Lf(\nu)= -\langle \frac{d}{dx} \partial f(\nu), \nu \rangle + \frac{1}{2} \langle \beta \partial^2 f(\nu), \nu^2 \rangle.
\end{align}
Note that in this case the function $\beta$  has to be a positive semi-definite kernel to guarantee that \eqref{eq:matrix} is satisfied.  It can then be interpreted as the covariance structure between the maturities. To see this, we compute  the carr\'e du champ operator associated with $L$, which is defined as the following symmetric bilinear map 
$$\Gamma(f,h)=L(fh)-fLh-hLf \; \; \; \; \; \; \; \;f,h\in \mathcal{D}.$$
This operator then yields a formula for the quadratic variation process of the martingale $N^f$ given by \eqref{martprob}.
Indeed, denote the quadratic variation process of $N^f$ by $\langle\langle N^f \rangle\rangle$. Then, e.g.~by adapting \cite[Proposition VIII.3.3]{RY:99} to the current setting we have
\begin{equation}\label{eq:quadraticvariationforward}
\langle\langle N^f \rangle\rangle_t=\int_0^t \Gamma(f,f)(\mu_s)ds,
\end{equation}
where $(\mu_t)_{t \in [0,T]}$ denotes the solution to the martingale problem for $L$ as of \eqref{eq:L-BS}.
A straightforward calculation then yields for
 $f\in \mathcal{D}$ of the form $f(\nu)=\Psi(\langle  \nu,g_1\rangle, \ldots,\langle  \nu,g_m\rangle)$ with $g_i \in D_1, \Psi \in C^\infty(\mathbb{R}^m)$ for some $m\in\mathbb{N}_0$,
$$\Gamma(f,f)(\nu)=\sum_{i,j=1}^m
\partial_i \Psi(\langle g_1,\nu\rangle,\ldots,\langle g_m,\nu\rangle)\partial_j \Psi(\langle g_1,\nu\rangle,\ldots,\langle g_m,\nu\rangle)\langle \beta g_i\otimes g_j,\nu^2 \rangle.$$
Hence, 
\begin{align}\label{eq:quadraticvar}
\langle\langle N^f \rangle\rangle_t=\int_0^t \sum_{i,j=1}^m
\partial_i \Psi(\langle g_1,\nu\rangle,\ldots,\langle g_m,\nu\rangle)\partial_j \Psi(\langle g_1,\nu\rangle,\ldots,\langle g_m,\nu\rangle)\langle \beta g_i\otimes g_j,\mu_s^2 \rangle ds.
\end{align}

Comparing this with the classical form of a multivariate Black-Scholes model of dimension $d$ 
\[
dS_t= \text{drift } dt + \diag(S_t) \sigma dW_t,
\]
where $\sigma  \in \mathbb{R}^{d \times d}$, $W$ is a $d$-dimensional Brownian motion and where
we deliberately do not specify the drift part, we see e.g.~by Ito's formula or \eqref{eq:quadraticvariationforward}  applied to the current situation that the quadratic variation of 
\[
h(S_t)- h(S_0)- \int_0^t \mathcal{A} h(S_s) ds, \quad h \in C^{\infty}(\mathbb{R}^d)
\]
is given by 
\[
\int_0^t \sum_{i,j}^d \partial_i h(S_s) \partial_j h(S_s) S_{s,i} S_{s,j} \sigma^2_{ij} ds.
\]
This is exactly the same form as \eqref{eq:quadraticvar}, which can be seen by setting $m=d$, identifying $\Psi$ with $h$, inserting $g_i= \delta_i$, noticing that $S$ takes the role of $\mu$ and replacing $\beta$ by $\sigma^2$.
\end{remark}

\subsection{Affine measure-valued HJM-models}
Another tractable subclass of polynomial processes are \emph{affine type} models, corresponding to Dawson-Watanabe-type superprocesses. There the focus will lie on the expression of the Laplace transform given in terms of one-dimensional non-linear PDEs of Riccati type.
For this affine subclass it will turn out that $Q_2 \equiv 0$. To introduce these models, 
we first recall the definition of affine type operators in sense of Definition~6.1 in \cite{CGPS:21}. Such definition is inspired by the classical form of the infinitesimal generator of affine processes on finite dimensional state spaces, see, e.g.,  \cite{KW:71, DFS:03, CFMT:11, CKMT:16}. Again, recall that the set $\Dcal$ is given by
$$
\mathcal{D}:=\operatorname{span}\{ P^D(M_+(E)), F^D_0  \}.
$$

\begin{definition}
We say that a linear operator  $L:\Dcal\to C(M_+(E))$ is of \emph{affine type} on $M_+(E)$ if there exist maps $F: D \to \mathbb{R}$ and $R: D \to C(E)$
such that 
\[
L \exp(\langle g, \fdot \rangle )(\nu)= (F(g) + \langle R(g), \nu \rangle) \exp(\langle g, \nu \rangle )
\]
for all $g \in D_-$ and $\nu \in M_+(E)$, where $D_-$ denotes all function in $D$ with values in $\mathbb{R}_-$.
\end{definition}

The following theorem now gives the precise form of  
affine type operators such that  the HJM-condition is satisfied \eqref{eq:martmeas}. It also states the well-posedness of the associated martingale problem with unique solution $(\mu_t)_{t \in [0,T]}$ and the exponential linear form of its Laplace transform which can be computed by solving a Riccati PDE.
In its formulation we use the notation $x^+:=\max\{x,0\}$.

\begin{theorem}\label{mainaffine}
Let $D=D_1$ and $L:\Dcal\to C(M_+(E))$ be the linear operator of form~\eqref{eq:op} given by
\begin{align}\label{eq:affine}
Lf(\nu)=-\langle \frac{d}{dx} \partial f (\nu), \nu \rangle + \frac{1}{2}  \langle\alpha \diag(\partial^2 f (\nu)), \nu \rangle,
\end{align}
for some $ \alpha \in C(E)$ taking values in $\mathbb{R}_+$.
Then $L$ is of affine type, its martingale problem
is well-posed and the   $M_+(E)$-valued solution $(\mu_t)_{t \in [0,T]}$ has continuous paths, satisfies the HJM-drift condition \eqref{eq:martmeas}. Moreover, the market corresponding to $(\mu_t)_{t\in[0,T]}$ satisfies NAFLVR in the
sense of Assumption~\ref{ass:NAFLVR}. 
Finally, the Riccati PDE
\[
\frac d {dt} \psi_t= -\frac{d}{dx} \psi_t + \frac{1}{2} \alpha \psi_t^2, \quad \psi_0 = g \in (D_1)_-
\]
admits the $(D_1)_-$-valued solution 
$$\psi_t(x)=\frac {g\big((x-t)^+\big)}{1-g\big((x-t)^+\big)\int_0^t\frac 1 2 \alpha\big((x-s)^+\big)ds}$$
and for all $0\leq s\leq t\leq T$ the conditional Laplace transform  has almost surely the following representation
\begin{equation}\label{eq:Laplacetrans}
\mathbb{E}\left[\exp( \langle g, \mu_t \rangle) \,\big|\, \mathcal{F}_s \right ]= \exp( \langle \psi_{t-s}, \mu_s \rangle), \quad  g \in (D_1)_-.
\end{equation}
\end{theorem}

\begin{proof}
A direct computation yields that $L$ is of affine type with $F=0$ and $$R(g)(x)=-\frac{d}{dx} g(x) + \frac{1}{2} \alpha(x) g^2(x),$$ 
for each $g\in (D_{1})_-$. Existence of solutions to the martingale problem and all related assertions follow from Theorem~\ref{main1} and the subsequent corollary. Uniqueness is a consequence of \cite[Corollary 6.9]{CGPS:21} by noticing that $- \frac{d}{dx}$ is the generator of a strongly continuous positive semigroup, namely the negative shift semigroup $(P_t)_{t \in [0,T]}$ given by $P_tf(x):=f((x-t)^+)$. 
Finally, a direct computation using that $g\in (D_1)_-$ shows that $\psi$ is a $(D_1)_-$-valued solution of the Riccati PDE. Since $\sup_{x,t\in[0,T]}|\psi_t(x)|<\infty$ and $\sup_{x,t\in[0,T]}|\frac d {dt}\psi_t(x)|<\infty$  we can conclude that it is also a solution in sense of Definition~6.4 in \cite{CGPS:21} and 
 the form of the Laplace transform follows from Theorem~6.6 in same paper.
\end{proof}

\begin{remark}
Notice that the operator $L$ defined in \eqref{eq:affine} is a variant of the Dawson-Watanabe superprocess
 (in the terminology of~\cite{E:00}\footnote{Note that in \cite{L:10}  ``Dawson-Watanabe superprocess'' is used for a class of measure-valued branching processes which can also exhibit jumps.}), also called super-Brownian motion,  where the constant diffusion coefficient is replaced by a function and where the spatial motion is governed by the negative shift semigroup.
It can thus be seen as 
measure-valued analog of one-dimensional Cox-Ingersoll-Ross processes.
\end{remark}

\subsection{Option pricing}\label{sec:options}

We now investigate option pricing in the two model classes considered above. We focus on European-style options on  future contracts with payoff function $\theta:\mathbb{R}\rightarrow\mathbb{R}$ and exercise time $\tau \in [0,\tau_1]$. Classical examples are standard call and put options with strike price $K\geq 0$, for which the payoff function is defined by $\theta(x)=(x-K)^+$, respectively by $\theta(x)=(K-x)^+$.
Computing the price of such an option at time  $t \leq \tau \leq \tau_1$ then amounts to calculate
\begin{align}\label{eq:pricing}
\pi_t=\mathbb{E}_{\mathbb{Q}}[ \theta(F(\tau, \tau_1, \tau_2))| \mathcal{F}_t]
\end{align}
under some equivalent (local) martingale measure $\mathbb{Q}$ (compare e.g.~\cite{BBK:08} or \cite{BDL:21}). We here shall suppose that $F(\tau, \tau_1, \tau_2)$ is given by
$$
F(\tau, \tau_1, \tau_2)=\int_E  w(\tau+x; \, \tau_1, \tau_2)\mathbbm{1}_{(\tau_1, \tau_2]}(\tau+x) \mu_\tau(dx),
$$
as in
\eqref{eq:future1}, with $(\mu_t)_{t \in [0,T]}$ being either a  Black-Scholes type or affine measure-valued HJM-model under $\mathbb{Q}$.

\subsubsection{Option pricing in Black-Scholes type measure-valued HJM-models}

As shown in Proposition \ref{prop:BStype} all assumptions such that the moment formula holds true are satisfied if $(\mu_t)_{t \in [0,T]}$ is a Black-Scholes type measure-valued HJM model. This can then be exploited to use polynomial techniques to compute \eqref{eq:pricing}. The simplest one is a variance reduction technique for  Monte Carlo pricing. In this case one approximates the true payoff $\theta$ via a polynomial $p$ (e.g. in the uniform norm on some interval) and then computes an estimator for the price (here for $t=0$) via
\[
\frac{1}{N}\sum_{i=1}^N \theta(F(\tau, \tau_1, \tau_2)(\omega_i))- c\Big( p(F(\tau, \tau_1, \tau_2)(\omega_i))- \mathbb{E}_{\mathbb{Q}}[ p(F(\tau, \tau_1, \tau_2))]\Big)
\]
for some constant $c \in \mathbb{R}$. By choosing $c$ optimally namely as
$$
c=\frac{\text{Cov}(\theta(F(\tau, \tau_1, \tau_2)), p(F(\tau, \tau_1, \tau_2)))}{\text{Var}(p(F(\tau, \tau_1, \tau_2)))},
$$
the variance with respect to the usual Monte Carlo estimator corresponding to the case $c=0$ is reduced by the factor $1- \text{Corr}^2(\theta(F(\tau, \tau_1, \tau_2)), p(F(\tau, \tau_1, \tau_2)))$. Suppose now that $p(x)=\sum_{k=0}^m   a_k x^k$ for coefficients $a_k \in \mathbb{R}$ and hence
$$\mathbb{E}_{\mathbb{Q}}[ p(F(\tau, \tau_1, \tau_2))]=\sum_{k=0}^m a_k\mathbb{E}_{\mathbb{Q}}[ F(\tau, \tau_1, \tau_2)^k].$$
As in the proof of Theorem~\ref{th:main1}, let $h_n\in C(\R^\infty)$ be a bounded sequence of  functions approximating $\mathbbm{1}_{(\tau_1, \tau_2]}$ pointwise such that $h_n(u)=0$ for each $u\in[0,\tau_1]$.
Then since by \cite[Lemma~5.7]{CGPS:21} all moments are finite, we have by dominated convergence 
\begin{align*}
\mathbb{E}_{\mathbb{Q}}[ F(\tau, \tau_1, \tau_2)^k]&= \mathbb{E}_{\mathbb{Q}}\bigg[ \Big( \int_{(\tau_1-\tau, \tau_2-\tau]}  w(\tau+x; \, \tau_1, \tau_2) \mu_{\tau}(dx)\Big)^k\bigg]\\
&= \mathbb{E}_{\mathbb{Q}}[ \langle  w(\tau+\fdot; \, \tau_1, \tau_2)\mathbbm{1}_{(\tau_1, \tau_2]}(\tau+ \fdot), \mu_{\tau} \rangle^k]\\
&=\mathbb{E}_{\mathbb{Q}}[ \lim_{n \to \infty} \langle  w(\tau+\fdot; \, \tau_1, \tau_2) h_n(\tau+ \fdot), \mu_{\tau} \rangle^k]\\
&=\lim_{n \to \infty}\mathbb{E}_{\mathbb{Q}}[  \langle  w(\tau+\fdot; \, \tau_1, \tau_2) h_n(\tau+ \fdot), \mu_{\tau} \rangle^k].
\end{align*}
As $(\mu_t)_{t \in [0,T]}$ is a Black-Scholes type measure-valued HJM-model this  expected value
can then be computed  by solving 
\[
\frac d {dt} \vec{g}^n_t =  L_m \vec{g}^n_t, \qquad 
\vec{g}^n_0=(g^n_{0,0}, g^n_{0,1}, \ldots, g^n_{0,m})m
\]
where  $L_m$ is given by \eqref{eq:Lm1} and 
\[
g^n_{0,k}= a_k (w(\tau+\fdot; \, \tau_1, \tau_2)h_n(\tau+ \fdot))^{\otimes k}, \quad k=0, \ldots,m.
\]
Alternatively to this variance reduction technique 
option pricing methods that build on orthogonal polynomial expansions as in \cite{AF:20} can be applied.

\subsubsection{Option pricing in affine measure-valued HJM-models}

Let us now turn our attention to affine measure-valued HJM-models. In this case the Laplace transform is explicitly known once the associated Riccati-PDE is solved. Since this can  be transferred to the Fourier-transform, Fourier option pricing methods (see e.g.~\cite{CM:99}) that are often applied for standard finite dimensional affine processes also come into play in this measure-valued situation.

Suppose $ \theta $ can be expressed by 
\begin{align}\label{eq: fourier inv}
\theta(x)=\int_{\mathbb{R}} \exp( (C+\text{i}\lambda) x) \widehat \theta(\lambda)d\lambda, \quad \lambda \in \mathbb{R},
\end{align}
for some integrable function $\widehat \theta$ and some constant $C\in \mathbb{R}$. If, moreover
\begin{align}\label{eq:expmoment}
\mathbb{E}_{\mathbb{Q}}[\exp{ (C F(\tau,\tau_1, \tau_2))} ]<\infty,
\end{align}
then \eqref{eq:pricing} for $t=0$ can be expressed as
\begin{align}
\pi_0&=\mathbb{E}_{\mathbb{Q}}\left[\int_{\mathbb{R}}\exp\Big((C+ \text{i} \lambda) F(\tau,\tau_1, \tau_2)\Big) \widehat \theta (\lambda)d\lambda \right]\nonumber \\
&=\int_{\mathbb{R}}\mathbb{E}_{\mathbb{Q}}\Big[\exp\Big(( C+i\lambda) F(\tau,\tau_1, \tau_2)\Big)\Big]\widehat \theta(\lambda)d\lambda \nonumber \\
&=\int_{\mathbb{R}}\mathbb{E}_{\mathbb{Q}}[\exp(\langle (C+\text{i}\lambda) w(\tau+\fdot; \, \tau_1, \tau_2)\mathbbm{1}_{(\tau_1, \tau_2]}(\tau+ \fdot), \mu_{\tau}(dx) \rangle )]\widehat \theta(\lambda)d\lambda. \label{eq:price}
\end{align}
As above let us consider a bounded sequence of smooth functions $h_n$ approximating  $\mathbbm{1}_{(\tau_1, \tau_2]}$. Then \eqref{eq:price} becomes
\begin{align*}
\pi_0&= 
\int_{\mathbb{R}}\mathbb{E}_{\mathbb{Q}}[\lim_{n \to \infty} \exp(\langle (C+\text{i}\lambda) w(\tau+\fdot; \, \tau_1, \tau_2) h_n(\tau+ \fdot), \mu_{\tau}(dx) \rangle )]\widehat \theta(\lambda)d\lambda \nonumber \\
&=\int_{\mathbb{R}}\lim_{n \to \infty}\mathbb{E}_{\mathbb{Q}}[ \exp(\langle (C+\text{i}\lambda) w(\tau+\fdot; \, \tau_1, \tau_2) h_n(\tau+ \fdot), \mu_{\tau}(dx) \rangle )]\widehat \theta(\lambda)d\lambda \nonumber\\
&=\int_{\mathbb{R}} \lim_{n \to \infty}\exp(\langle \psi_{\tau}^{n, \lambda}, \mu_0 \rangle)\widehat \theta(\lambda)d\lambda\nonumber\\
&=\int_{\mathbb{R}}\exp(\langle \psi_{\tau}^{ \lambda}, \mu_0 \rangle)\widehat \theta(\lambda)d\lambda,
\end{align*}
where the second equality follows due to  \eqref{eq:expmoment} from dominated convergence  and the third is a consequence of 
\eqref{eq:Laplacetrans} extended to the Fourier-Laplace transform (see Remark \ref{rem:extension} below). Here $\psi^{\lambda}$ solves the Riccati PDE with initial condition  $g_\tau^\lambda(x):=(C+\text{i}\lambda) w(\tau+x; \, \tau_1, \tau_2)\mathbbm{1}_{(\tau_1, \tau_2]}(\tau+ x).$ Its value at time $t$ is thus given by
\begin{align*}
    \psi_t^\lambda(x)&=\frac {g_\tau^\lambda\big((x-t)^+\big)}{1-g_\tau^\lambda\big((x-t)^+\big)\int_0^t\frac 1 2 \alpha\big((x-s)^+\big)ds}\\
&=
\frac {(C+\text{i}\lambda) w(\tau+(x-t)^+; \, \tau_1, \tau_2)}
{1-(C+\text{i}\lambda) w(\tau+(x-t)^+; \, \tau_1, \tau_2)\int_0^t\frac 1 2 \alpha\big(\tau+(x-s)^+\big)ds}
\mathbbm{1}_{(\tau_1, \tau_2]}(\tau+(x-t)^+).
\end{align*}
For $t=\tau$ this expression simplifies to
\begin{align}\label{eq:explicit}
    \psi_\tau^\lambda(x)&=
\frac {2(C+\text{i}\lambda) w(x; \, \tau_1, \tau_2)}
{2-(C+\text{i}\lambda) w(x; \, \tau_1, \tau_2)(A(x)-A(x-\tau))}
\mathbbm{1}_{(\tau_1, \tau_2]}(x),
\end{align}
where $A$ is the primitive of $\alpha$.
\begin{remark}\label{rem:extension}
Even if $g_\tau^\lambda(x)$ is not non-positive, we can  follow the classical affine literature (see  \cite{KM:15}) using the exponential moment condition \eqref{eq:expmoment} to extend the results of Theorem~\ref{mainaffine} to $g_\tau^\lambda$.
\end{remark}

\begin{example} \label{ex1}
From Fourier analysis we know that for $C>0$ it holds
$$y^+=\frac 1 {2\pi}\int_\R \frac{e^{(C+i\lambda)y}}{ (C + \text{i} \lambda)^2}d\lambda.$$
We can thus see that $\theta(x):=(x-K)^+$ satisfies \eqref{eq: fourier inv} for
\[
\widehat{\theta}(\lambda)= \frac{\exp((-C-\text{i}\lambda)K)}{2 \pi (C + \text{i} \lambda)^2}
\]
so that we obtain using \eqref{eq:explicit}
\begin{align*}
    &\E_\Q[(F(\tau;\tau_1,\tau_2)-K)^+]
=\int_{\mathbb{R}}\exp(\langle \psi_{\tau}^{ \lambda}, \mu_0 \rangle)\widehat \theta(\lambda)d\lambda\\
&\quad=\frac 1 {2\pi}\int_\R \frac{1}{(C + \text{i} \lambda)^2}
e^{(-C-\text{i}\lambda)\big(K-\int_{(\tau_1, \tau_2]} \frac {2 w(x; \, \tau_1, \tau_2)}
{2-(C+\text{i}\lambda) w(x; \, \tau_1, \tau_2)(A(x)-A(x-\tau))} \mu_0(dx)\big)} d\lambda.
\end{align*}
 In the particular case of $w(x;\tau_1,\tau_2):=(\tau_2-\tau_1)^{-1}$ this expression simplifies to
\begin{align}\label{eq:explicit_fourier}
\frac 1 {2\pi}\int_\R \frac{1}{(C + \text{i} \lambda)^2}
e^{(-C-\text{i}\lambda)\big(K-\int_{(\tau_1, \tau_2]}\frac {2 }
{2(\tau_2-\tau_1)-(C+\text{i}\lambda) (A(x)-A(x-\tau))} \mu_0(dx)\big)} d\lambda.
\end{align}
\end{example}

\subsection{Calibration to market option prices}\label{sec:cali}

The above pricing methods clearly offer tractable  calibration procedures, where the
function-valued characteristics of the polynomial models
can be calibrated to match the market option prices.

Indeed, the goal is to parametrize these functions as neural networks to get \emph{neural measure-valued processes}, an analog to neural SDEs and SPDEs (see e.g. \cite{GSSSZ:20, CKT:20, SL:21}).
Note that we could of course also  parametrize the map $M_+(E) \to \mathbb{R}: \nu \to Q(g ,\nu)$ for $g \in  D \otimes D$, corresponding to the diffusion part in \eqref{eq:diff} via neural networks taking measures as inputs.
We refer to \cite{BDG:21, AKP:22, CST:22} for neural networks with infinite dimensional inputs (and also outputs). 

By considering the polynomial class the modeling simplifies to standard neural networks that parameterize the function-valued characteristics. While tractability is preserved, this leads to potentially highly parametric infinite dimensional stochastic models which have a great potential to capture multifarious features of energy markets. Parameter optimization 
can be performed efficiently due to
stochastic gradient descent methods and
 tractable backpropagation schemes.

We shall now exemplify this by means of an affine measure-valued HJM model whose generator is given by \eqref{eq:affine} with $\alpha \in C(E)$ being  parameterized by a non-negative neural network. We here do not aim to perform a full calibration to all available option price data, but just show a proof of concept by calibrating the model to call option prices for one specific maturity and one delivery period (i.e.~one slice of the volatility surface (or cube)). 

We use  EEX German Power data extracted from
$$\text{\url{https://www.eex.com/en/market-data/power/options}}$$
at March 22, 2022 and calibrate to call options with maturity April 26, 2022 and delivery period one month, starting on May 1, 2022.
As our goal is to provide just a calibration example, we use a simple $L^1$-criterion for minimizing the distance between the market call option prices and the respective model prices, i.e. we minimize
\begin{align}\label{eq:minimization}
\sum_{K} |\pi_{\text{mkt}}(K)-\pi_0(K)|,
\end{align}
where $\pi_{\text{mkt}}(K)$ denotes the market call option price with strike $K$ and $\pi_0(K)$ the model call option price obtained via Fourier pricing as explained in Example~\ref{ex1} for
$$w(u;\tau_1;\tau_2):=\frac 1 {\tau_2-\tau_1}.$$
The day ahead forward price on March 22, 2022 is 236.49 and we normalize all prices by this quantity to start with $1$. We then deal with 
10 market strikes ranging from $0.9$ to $1.1$ 
and 71 daily values of the initial forward curve, encoded in $\mu_0$, from March 22, 2022 to May 31, 2022.
The corresponding $\tau_1$ and $\tau_2$ are thus approximately given by $\tau_1=1.33/12$ and $\tau_2=2.33/12$. 
The integral with respect to $\lambda$  in \eqref{eq:explicit_fourier} is computed by discretization from $-100$ to $100$. The function $\alpha$ is parameterized by a $3$ layer neural network with activation functions $\tanh$ for the first layer and relu for the two outer layers. 

Computing then \eqref{eq:minimization} yields the calibration results illustrated in Figure \ref{fig:1}. Absolute, relative errors and also squared errors are shown in Figure \ref{fig:2}. Note that the calibration is extremely accurate for the at the money options. The mean absolute error (over all strikes) is 0.00096, the mean relative error 0.011 and the mean squared error $1.20 \times 10^{-6}$.
When the parameters of the neural networks are appropriately initialized, then training is very fast (some seconds on a standard laptop), but the initialization is quite crucial in order to start in regimes where the model prices are not too far from the the market prices. In Figure \ref{fig:3}, we show how the curves looked at the start of the training. The target curve can then be reached in around 150-300 gradient step iterations with a learning rate of 0.01.

\begin{figure} [ht!]
\begin{center}
\includegraphics[scale=0.5]{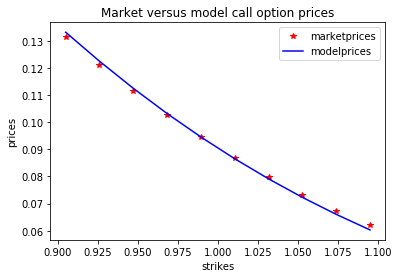}
\end{center}
\caption{Market versus model call option prices for options expiring on April 26, 2022, written on forward contracts with delivery during May.}
\label{fig:1}
\end{figure}

\begin{figure} 
\begin{center}
\includegraphics[scale=0.35]{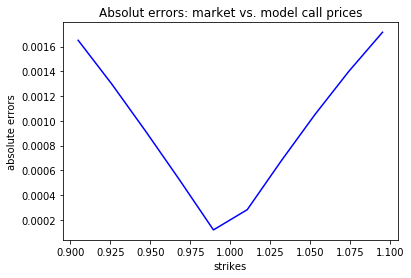}
\includegraphics[scale=0.35]{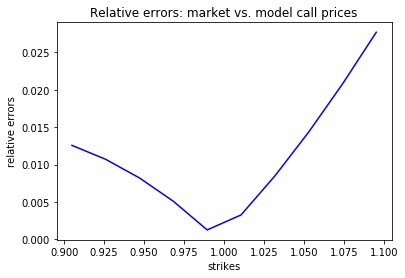}
\includegraphics[scale=0.35]{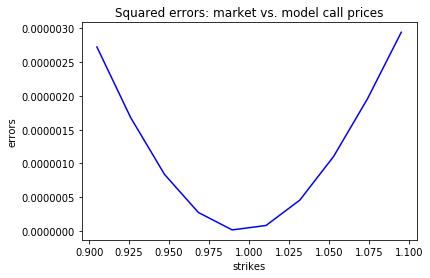}
\end{center}

\caption{Absolute, relative and squared errors between market and model call option prices.}\label{fig:2}
\end{figure}

\begin{figure} [ht!]
\begin{center}
\includegraphics[scale=0.5]{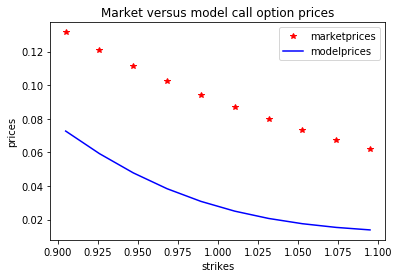}
\end{center}
\caption{Market versus model call option prices at the start of the training.}
\label{fig:3}
\end{figure}

As mentioned in Remark \ref{rem:inhomogen} the model can be easily extended by making the function-valued characteristics time-dependent, which then does not only allow to capture seasonality features but also to adapt the current calibration procedure to a slicewise calibration as conducted in \cite{CKT:20, GSSSZ:20, CGS:22}, leading to similarly accurate calibration results for all maturities.

\section{Measure-valued processes for renewable energy production} \label{sec:wind}

In this section we give a brief outlook how measure-valued processes could 
be used to model renewable energy production. We outline this by means of the  example of wind energy markets. We briefly introduce some characteristics of this market relying on \cite{BPL:18}. 

The production of wind energy has played an increasingly important role as a renewable energy source. This type of energy can be generated only in the presence of a suitable amount of wind: a minimum wind's speed is required and too strong wind would cause  damages. 
A direct law relating the maximal power production from a wind power plant to the wind speed is  Betz's law. If one introduces $m\in\mathbb{R}_+$ as the minimum wind speed which allows for  production and $M\in\mathbb{R}_+$ as the maximum one, Betz's law expresses the upper bound of wind power production $P^{upper}$ in terms of the wind speed $v$ at time $t$ as
\begin{equation}\label{eq:betz}
P^{upper}_t=hv^3_t\mathbbm{1}_{[m,M]}(v_t), \quad t \in [0,T],
\end{equation}
where $h$ is a positive constant describing the heat rate,  being a real positive constant measuring the efficiency of the plant (see \cite{BPL:18} for  further details). It can happen that the wind is too strong so that  the company faces a surplus of production that must be sold in the market. This can in turn trigger a decrease of related electricity prices. Hence it becomes essential for an energy producer to hedge against these risks. In energy markets there exist weather derivatives that have been designed to cover such kinds of exposure. A particular class of contracts  are  so-called quanto options that simultaneously take into account the power production driven by the wind speed, as well as the electricity price. In order to price such  contracts, one thus has to specify  dynamics for both the electricity price and the wind speed. 
According to \cite{BPL:18}, a typical payoff of a quanto option is the product of two put options written on the one hand on electricity futures and on the other hand on the 
future cumulative gap between the theoretical total maximum production $P^{upper}$  and the actual wind power production $P$.
The price of such a quanto option can thus be written as
$$
\Pi_t=\mathbb{E}_{\mathbb{Q}}[(K_E-F_E(\tau;\tau_1,\tau_2))^+ (K_I-F_I(\tau_1,\tau_2))^+|\mathcal{F}_t], \quad t \leq \tau \leq \tau_1,
$$
where $F_E(\tau;\tau_1,\tau_2)$ is a future on electricity as in the previous sections, $F_I(\tau_1,\tau_2)$ is the described payoff  related to the wind power production given by
\begin{equation}\label{eq:winforward}
F_I(\tau_1,\tau_2)=\int_{\tau_1}^{\tau_2}(P_u^{upper}-P_u )du
\end{equation}
and $K_E,K_I>0$ are the respective strike values. 
Note that, when fixing $m=0$ and $M=\infty$ in \eqref{eq:betz}, \eqref{eq:winforward} describes the gap between the maximal wind power production and the actual one.

We now outline a possible method to exploit measure-valued processes for modeling  wind energy markets. Suppose that for each time $t$ and each location $x$ the wind's speed for the time period $[t,t+24h]$ is known at time $t$. Let $(\nu_t(x,\fdot))_{t\geq 0}$ be a measure-valued process on $M_+(\mathbb{R}_+)$ and interpret $\nu_t(x,\fdot)$ as the normalized distribution of the wind speed $v \in \mathbb{R}_+$ at the geographical location $x \in \mathbb{R}^2$ over 24 hours starting at time $t$. For example, if at day $t$ the forecasts for the next 24 hours predict a constant wind of speed $v_0$ at location $x$ we set $\nu_t(x,\fdot)=\delta_{v_0}$.

Assuming that the wind power production over 24 hours of a  unit located in $x$ is a deterministic function $w$ of the wind speed and the location of the unit, then  we can express the average wind's power's production $P_t(x)$ at time $t$ by
$$
P_t(x)=\int_{\mathbb{R}_+} w(v,x)\nu_t(x,dv).
$$
According to Betz's law this quantity is bounded by 
$P_t(x)^{upper}=h(x)\int_{\R_+}v^3 \nu_u(x,dv)$.
Letting $\lambda$ be the distribution of power plants' units on the territory, \eqref{eq:winforward} then becomes
\begin{equation*}
F_I(\tau_1,\tau_2)=\int_{\tau_1}^{\tau_2}\int_{\R^2}\int_{\R_+}\big(h(x)v^3 - w(v,x)\big)\nu_t(x,dv)\lambda(dx)dt.
\end{equation*}
One possibility would then be to model $\mu_t(dv,dx):=\nu_t(x,dv)\lambda(dx)$ directly with a measure-valued model. 
By adding a further component to the underlying space $\mathbb{R}_+\times \mathbb{R}^2$ one can then simultaneously model also the electricity futures via the measure-valued approach.
Similar tractable methods as those studied in the previous sections can therefore be used to price quanto options.

\appendix
\section{Auxiliary results}
The following lemma is a key result for the proof of Theorem~\ref{th:main1}.
\begin{lemma}\label{lem:mart}
Let $(\mu_t)_{t \in [0,T]}$ be a $M_+(E)$-valued process such that $\mathbb{E}_{\mathbb{Q}}[\sup_{t \in [0,T]} \mu_t(E)] < \infty$ and such that
for all $\phi \in D_1$ 
\begin{align}\label{eq:driftmu}
\langle \phi, \mu\rangle + \int_0^\cdot \langle \phi', \mu_s \rangle ds, 
\end{align}
is a local martingale on $[0,T]$. Then 
for all $\Phi \in D_2(\tau_1) $ 
\begin{align}\label{eq:martingalecond}
\langle \Phi(t, \fdot), \mu_t \rangle - \langle \Phi(0, \fdot), \mu_0 \rangle + \int_0^{t}( \langle \frac d {dx}  \Phi(s, \fdot), \mu_s \rangle - \langle \frac d {dt} \Phi(s, \fdot), \mu_s \rangle )ds,
\end{align}
is a martingale on $[0,\tau_1]$.
\end{lemma}

\begin{proof}
We first prove the assertions for functions $\Phi$ of the form $\Phi(t,x)=\psi(t) \phi(x)$ with 
$\psi \in C^{\infty}([0,\tau_1])$ 
and $\phi \in D_1$.
Denote by $Z$ the process
\[
Z_t=\langle \phi, \mu_t \rangle + \int_0^t \langle \phi', \mu_s \rangle ds, \quad t \in [0,T],
\]
which is a local martingale by \eqref{eq:driftmu}.
Then by the It\^o's product rule, we have
\begin{align}\label{eqn1}
    \psi(t)\langle \phi, \mu_t \rangle=
    \psi(0)\langle \phi, \mu_0 \rangle+ \int_0^t \psi'(s) 
\langle \phi, \mu_s \rangle -\psi(s) \langle \phi', \mu_s \rangle
ds + \int_0^t \psi(s) dZ_s.
\end{align}

Note then that it holds
\begin{align*}
\mathbb{E}\left[\left|\sup_{t \leq T} \int_0^t \psi(s) dZ_s\right|\right] &\leq 
\mathbb{E}\left[\sup_{t \leq T}|\psi(t)\langle \phi, \mu_t \rangle|\right] + \mathbb{E}\left[|\psi(0)\langle \phi, \mu_0 \rangle| \right] \\
&\quad + \mathbb{E}\left[\sup_{t \leq T} \left|\int_0^t \psi'(s) 
\langle \phi, \mu_s \rangle - \psi(s) \langle \phi', \mu_s \rangle ds\right|\right]\\
&\leq C \mathbb{E}\left[\sup_{t \leq T} \mu_t(E)\right] < \infty,
\end{align*}
where $C$ denotes some constant and where the  last inequality follows from the fact that $\phi \in D_1$ and $\psi \in  C^{\infty}([0, \tau_1])$. Since the last expression is finite by assumption, we can conclude that 
$\int_0^t \psi(s) dZ_s$ is a true martingale. Since
\begin{equation*}
    \Psi_x=\psi\phi'\text{\qquad and \qquad}\Psi_t=\psi'\phi,
\end{equation*}
the martingale condition for \eqref{eq:martingalecond} in for $\Phi(t,x)=\psi(t) \phi(x)$ follows from \eqref{eqn1}.

Next, note that the set of functions $\{ (t,x)  \mapsto \psi(t)\phi(x) \,| \,  \phi \in D_1, \psi \in C^{\infty}([0, \tau_1]) \}$ forms a point-separating algebra that vanishes nowhere,
and which has additionally the property that for all $v \in \mathbb{R}^2 \setminus \{0\}$ and $(t,x) \in [0, \tau_1] \times E$ such that $x \neq 0$ there exists some $\psi, \phi$ with $(\phi'(x),\psi'(t))v \neq 0$.
By an adaptation of Nachbin's version of the Stone-Weierstrass theorem (see \cite{N:49}) that it is dense
in $D_2(\tau_1)$ with respect to the $C^1$-norm
\[
\| \Phi\|_{C^1} :=
\sup_{(t,x) \in E \times E} | \Phi(t,x) |
+
\sup_{(t,x) \in E \times E} | \frac d {dt}\Phi(t,x) |
+
\sup_{(t,x) \in E \times E} | \frac d {dx} \Phi(t,x) |.
\]

Consider now a sequence $\Phi^n=\psi_n\phi_n$ converging in the $C^1$-norm to $\Phi$. Observe that this in particular implies $\sup_n\|\Phi^n\|_{C^1}<\infty$. Denote then by $M^n$ the true martingale
\[
M_t^n=\langle \Phi^n(t, \fdot), \mu_t \rangle - \langle \Phi^n(t, \fdot), \mu_0 \rangle + \int_0^t( \langle \frac d {dx}  \Phi^n(s, \fdot), \mu_s \rangle - \langle \frac d {dt}\Phi^n(s, \fdot), \mu_s \rangle )ds
\]
and by $M$ the corresponding right hand side without $n$.
By dominated convergence we know that
$\lim_{n\to\infty}M^n_t = M_t$ for all $t$ almost surely and
$$\sup_n|M^n_t|\leq 2(1+T)\sup_n\|\Phi^n\|_{C^1}\sup_{s\in[0,T]}\mu_s(E),$$
which is an integrable random variable. By the dominated convergence theorem we thus get
\[
\mathbb{E}[M_t | \mathcal{F}_s]=\mathbb{E}[\lim_{n \to \infty} M_t^n | \mathcal{F}_s]=\lim_{n \to \infty}  \mathbb{E}[M_t^n | \mathcal{F}_s]= \lim_{n \to \infty} M_s^n= M_s
\]
and $M_t\leq 2(1+T)\sup_n\|\Phi^n\|_{C^1}\sup_{s\in[0,T]}\mu_s(E)$ for each $t\in[0,T]$.
This implies that $M$ and thus \eqref{eq:martingalecond}  is a true martingale showing the claim.
\end{proof}

\section{Existence for martingale problems} 

The next lemma is an adaptation of a well-known result from \cite{EK:09} giving sufficient conditions for the existence of a solution to the martingale problem for a linear operator $L$. Here, it is crucial that $L$ acts of $C_0$-functions on a locally compact, separable, metrizable space, which is satisfied for $M_+(E)$ since $E$ is a compact set.
For the formulation of the lemma recall the set of functions $F^D_c$  defined in \eqref{F^D_c}.

\begin{lemma}\label{lem:martingaleprob}
Suppose that $L:F^D_c\to C_0(M_+(E))$ is of form \eqref{eq:diff} for some $B$ such that 
\begin{align}\label{eq:growth}
|B(1, \nu)| \leq C (1 + \langle 1, \nu \rangle) 
\end{align}
 for some constant $C >0$. Assume furthermore that $L$
 satisfies the positive maximum principle on $M_+(E)$. 
Then for every initial condition in $M_+(E)$, there exists a  continuous $M_+(E)$-valued solution to the martingale problem for $L:F^D_c\to C_0(M_+(E))$. Moreover, $(M_t)_{t\in[0,T]}$ for
$$M_t:=\mu_t(E)-\mu_0(E)-\int_0^t B(1,\mu_s)ds$$
defines a local martingale.
\end{lemma}

\begin{proof}
The proof is similar to the proof of \cite[Lemma~B.2(i)]{CGPS:21}, but we state it here for completeness.
We shall verify the  conditions of  \cite[Theorem~4.5.4]{EK:09}.
As already mentioned, $M_+(E)$ is a  locally compact, separable and metrizable space. Moreover, by \cite[Lemma~2.6]{CGPS:21}, we have that
$F^D_c$ 
as defined in \eqref{F^D_c} is a dense subset of $C_0(M_+(E))$. 
 Moreover, the positive maximum principle yields that $Lf|_{M_+(E)}=Lh|_{M_+(E)}$ for all $f,h\in F^D_c$  such that $f|_{M_+(E)}=h|_{M_+(E)}$. Since $L(F^D_c)\subseteq C_0(M_+(E))$ by assumption  we may
 regard $L|_{F^D_c}$ as an operator  on $C_0(M_+(E))$. 
This means that all the assumptions of in \cite[Theorem~4.5.4]{EK:09} are satisfied. Recall from Section \ref{sec:notation}  that $M^{\mathfrak{\Delta}}_+(E)$ denotes the one-point compactification of $M_+(E)$ and
define according to the same theorem the linear operator $L^{\mathfrak{\Delta}}$ on $C(M^{\mathfrak{\Delta}}_+(E))$ by 
\[
L^{\mathfrak{\Delta}}f|_{M_+(E)}=L((f-f(\mathfrak{\Delta}))|_{M_+(E)}), \quad L^{\mathfrak{\Delta}}f(\mathfrak{\Delta})=0
\]
for all $f \in C(M^{\mathfrak{\Delta}}_+(E))$ such that $(f-f(\mathfrak{\Delta})_{M_+(E)}) \in F^D_c$.

Then \cite[Theorem~4.5.4]{EK:09} yields that for every initial condition in $M^{\mathfrak{\Delta}}_+(E)$, 
there exists a solution $(\mu_t)_{t\in[0,T]}$ to the martingale problem for $L^{\mathfrak{\Delta}}$ with càdlàg sample paths  taking values in $M^{\mathfrak{\Delta}}_+(E)$.
Indeed, we obtain that
\eqref{martprob} is a bounded local martingale (and thus a true martingale) for each $f \in F^D_c$ where $L$ is replaced by $L^{\mathfrak{\Delta}}$. Moreover, by Proposition 2 in \cite{BE:85} we also know that $t\mapsto f(\mu_t)$ is continuous for each $f \in F^D_c$.

\

Fix now $n\geq\langle 1,\mu_0\rangle$, define  $\tau_n:=\inf\{t>0\colon\langle 1,\mu_t\rangle>n\}$ and  set $\tau_{ \mathfrak{\Delta}}:=\lim_{n\to\infty}\tau_n$. We now aim to show that for all $t \in [0,T] $ we have that $\P(\tau_{ \mathfrak{\Delta}}>t)=1$, showing that $\P(\mu_t\neq  \mathfrak{\Delta})=1$.
Fix $m>n$ and consider a map $\phi_m\in C_c^\infty(\R)$ such that $\phi_m(x)=x$ for $|x|\leq m$. Set
   $f(\nu):=\langle 1,\nu\rangle$ and $f_m:=\phi_m\circ f$. By continuity of $t\mapsto f_m(\mu_t)$, we already know that
  $$\tau_n<\tau_{ \mathfrak{\Delta}}\qquad\text{ and }\qquad\langle 1,\mu_{\tau_n\land t}\rangle\leq n$$
   for each $t\in[0,T]$. Observe then that
$\partial f_m(\nu)=\partial f(\nu)\equiv 1$ and $\partial^2 f_m(\nu)=\partial^2 f(\nu) \equiv 0$ for each $\nu\in M_+(E)$ such that $\langle 1, \nu\rangle\leq m$. This implies together with  linearity of $Q$ in its first argument   that
$$
M_{\mu_n\land t}
=f_m(\mu_{\tau_n\land t})-f_m(\mu_0)-\int_0^{\tau_n\land t}  Lf_m(\mu_s)ds
$$
which is a  local martingale. Since it is also bounded due to to condition \eqref{eq:growth}, we can conclude that it is a true martingale. 

Fatou's lemma  then yields
$$\E[\lim_{n\to\infty}f(\mu_{\tau_n\land t})  ]
\leq\lim_{n\to\infty}\E[f(\mu_{\tau_n\land t})
]
\leq \lim_{n\to\infty}\bigg(f(\mu_{0})+C\int_0^{ t}1+  \E[f(\mu_{\tau_n\land s})] ds\bigg).$$
By the Gronwall inequality we can thus conclude that
$$\E[\lim_{n\to\infty}n\mathbbm{1}_{\{\tau_n\leq t\}}+\lim_{n\to\infty}\langle 1,\mu_t  \rangle\mathbbm{1}_{\{\tau_n> t\}}]
=\E[\lim_{n\to\infty}f(\mu_{\tau_n\land t})
]\leq (f(\mu_0)+Ct)\exp(Ct)$$
and hence that $\P(\tau_{ \mathfrak{\Delta}}\leq t)\leq \lim_{n\to\infty}\P(\tau_n\leq t)=0$. This proves that every solution to the martingale problem takes values in $M_+(E)$ and hence the assertion. 

\end{proof}


\end{document}